%% file: 00_arxiv_swap_protocols_single_hashlock_contracts.tex
\documentclass[11pt]{article}

\input{packages}

\input{macros}

\title{On HTLC-Based Protocols for Multi-Party Cross-Chain Swaps\thanks{Research supported by NSF grant CCF-2153723.}}

\author[1]{Emily Clark}
\author[1]{Chloe Georgiou}
\author[1]{Katelyn Poon}
\author[1]{Marek Chrobak}

\affil[1]{University of California at Riverside}

\begin{document}

\maketitle

\begin{abstract}
In his 2018 paper, Herlihy introduced an atomic protocol for multi-party asset swaps across 
different blockchains. His model represents an asset swap by a directed graph whose 
nodes are the participating parties and edges represent asset transfers, and
rational behavior of the participants is captured by a preference relation between a protocol's outcomes. 
Asset transfers between parties are achieved using smart contracts.
These smart contracts are quite intricate and involve solving computationally intensive graph problems,
limiting practical significance of his protocol. 
His paper also describes a different protocol
that uses only standard hash time-lock contracts (HTLC's), but this simpler protocol applies 
only to some special types of digraphs. He left open the question whether there is
a simple and efficient protocol for cross-chain asset swaps in arbitrary digraphs. Motivated by this open problem,
we conducted a comprehensive study of \emph{HTLC-based protocols}, in which all asset transfers are implemented with HTLCs.
Our main contribution is a full characterization of swap digraphs that have such protocols.
\end{abstract}


\section{Introduction}
\label{sec: introduction}
\input{01_introduction.tex}


\section{Graph Terminology}
\label{sec: graph terminology}
\input{02_graph_terminology.tex}


\section{Multi-Party Asset Swaps}
\label{sec: multi-party asset swaps}
\input{03_multi-party_asset_swaps.tex}


\section{An Atomic Protocol for Reuniclus Digraphs}
\label{sec: an atomic protocol for reuniclus digraphs}
\input{04_atomic_protocol_for_reuniclus_digraphs.tex}


\section{A Characterization of Digraphs that Admit {\SingleHashLock}-Based Protocols}
\label{sec: characterization of digraphs that admit shl protocols}
\input{05_characterization_of_digraphs_that_admit_shl_protocols.tex}


\bibliographystyle{plain}
\bibliography{chain_swaps_references.bib}




\end{document}

%% file: packages.tex
\usepackage[margin=2.5cm]{geometry}
\usepackage{authblk}

\usepackage{amssymb,amsmath,amsthm,amsfonts,thmtools}
\usepackage{color}
\usepackage[export]{adjustbox}
\usepackage{yfonts}
\usepackage{mathrsfs}

\usepackage{graphicx}
\usepackage{xspace}

\usepackage{times}
\usepackage{verbatim}
\usepackage{enumitem}
\usepackage{graphicx,wrapfig,lipsum}
\usepackage[font=small]{caption}
\usepackage{subcaption}
\usepackage[english]{babel}
\usepackage[utf8]{inputenc}

\usepackage[noend]{algpseudocode}
\usepackage{url}
\usepackage[all]{xy}
\usepackage{rotating}
\usepackage{ifpdf}
\usepackage{tcolorbox}
\usepackage{lipsum}

\usepackage{mdframed}             
\usepackage{xifthen} 
\usepackage{mathtools}  

\DeclareMathAlphabet\mathbfcal{OMS}{cmsy}{b}{n}

\usepackage{comment}
\usepackage[T1]{fontenc}

\usepackage[hidelinks]{hyperref} 
\usepackage{cleveref}

\usepackage[section]{placeins}
\usepackage{float}
\usepackage[title]{appendix}
\usepackage{graphicx}
\usepackage{tikz}
\usetikzlibrary{automata,positioning}
\usepackage{dsfont}
\usepackage{xspace}
\usepackage{enumitem}
\setdescription{font=\normalfont}

%% file: macros.tex
\colorlet{darkgreen}{green!45!black}

\newcommand{\ProtBDP}{{\textsf{BDP}}}
\newcommand{\ProtRDP}{{\textsf{RDP}}}

\newcommand{\protP}{{\mathbb{P}}}
\newcommand{\leader}{\ell}

\newcommand{\SingleHashLock}{{\text{SHL}}}

\newcommand{\outcome}{{\omega}}
\newcommand{\outcomepair}[2]{\angled{#1\,|\,#2}}

\newcommand{\Deal}{\textsc{Deal}\xspace}
\newcommand{\Dealv}[1]{{\textsc{Deal}_{#1}}\xspace}

\newcommand{\NoDeal}{\textsc{NoDeal}\xspace}
\newcommand{\NoDealv}[1]{{\textsc{NoDeal}_{#1}}\xspace}

\newcommand{\Discount}{\textsc{Discount}\xspace}
\newcommand{\Discountv}[1]{{\textsc{Discount}_{#1}}\xspace}

\newcommand{\FreeRide}{\textsc{FreeRide}\xspace}
\newcommand{\FreeRidev}[1]{{\textsc{FreeRide}_{#1}}\xspace}

\newcommand{\Underwater}{\textsc{Underwater}\xspace}

\newtheorem{theorem}{Theorem}
\newtheorem{lemma}{Lemma}
\newtheorem{corollary}[lemma]{Corollary}
\newtheorem{claim}[lemma]{Claim}

\newenvironment{theoremrestated}[1]{\medskip\noindent\textbf{Theorem~#1.}\it\ }
{\medskip}

\newcommand{\ignore}[1]{}

\newcommand{\margincomment}[2]%
{\marginpar{\footnotesize\raggedright {\color{red}#1}: #2}}

\newcommand{\myparagraph}[1]{{\medskip\noindent\textbf{#1}}}

\newcommand{\mycase}[1]{{\underline{Case~#1}:}}

\newcommand{\braced}[1]{{ \left\{ {#1} \right\} }}

\newcommand{\angled}[1]{{ \left\langle {#1} \right\rangle }}

\newcommand{\suchthat}{{\;:\;}}

\newcommand{\onestep}{{\textbullet}}

\newcommand{\NP}{\ensuremath{\mathbb{NP}}\xspace}

\newcommand{\controlgraph}{\cal{K}}

\newcommand{\inNeighbors}[1]{{N_{#1}^{in}}}
\newcommand{\outNeighbors}[1]{{N_{#1}^{out}}}
\newcommand{\inArcs}[1]{{A_{#1}^{in}}}
\newcommand{\outArcs}[1]{{A_{#1}^{out}}}
\newcommand{\outcomein}[1]{{\outcome_{#1}^{in}}}
\newcommand{\outcomeout}[1]{{\outcome_{#1}^{out}}}

\newcommand{\maxDfromLeader}[1]{{D^-_{#1}}}
\newcommand{\maxDfromSubleader}[1]{{B^-_{#1}}}
\newcommand{\maxDtoLeader}[1]{{D^+_{#1}}}
\newcommand{\ContractCreationTime}{{D^\ast}}
\newcommand{\RContractCreationTime}{{B^\ast}}

\newcommand{\ListLengths}{\setlength{\itemsep}{0ex}\setlength{\topsep}{1ex}\setlength{\partopsep}{0ex}}



%% file: 01_introduction.tex


In his 2018 paper~\cite{Herlihy18}, Herlihy introduced a model for multi-party asset swaps across different blockchains. 
In this model an asset swap is represented by a strongly connected directed graph (digraph), with each vertex corresponding to one 
party and each arc representing
a pre-arranged asset transfer between two parties. Given such a digraph, the goal is to design a protocol to implement
the transfer of all assets. There is no guarantee that all parties will follow the protocol --- some parties may attempt to cheat, even
colluding between themselves, or behave irrationally. The protocol must guarantee, irrespective of the
behavior of other parties, that
each honest party will end up with an outcome that it considers acceptable. The protocol should also discourage
cheating, so that any coalition of parties cannot improve its outcome by deviating from the protocol.
These two conditions are called \emph{safety} and \emph{strong Nash equilibrium}, respectively.
A protocol that satisfies these conditions is called \emph{atomic} (see the formal definitions in Section~\ref{sec: multi-party asset swaps}).

In this model, Herlihy~\cite{Herlihy18} developed an atomic protocol for multi-party asset swaps
in arbitrary strongly connected digraphs. 
This result constitutes a major advance in our understanding of fundamental principles of asset swaps, as 
it shows that it is possible, in principle, 
to securely implement complex transactions involving multiple parties even if their assets are hosted on multiple blockchains. 
Yet the practical significance of his protocol is limited due to its intricacy and high computational complexity.
Asset transfers in this protocol are realized via smart contracts that need to store and process
the whole swap digraph, and designing these contracts involves solving $\NP$-hard graph problems, like
finding long paths or a small feedback vertex set.
To achieve the  safety property, it uses an elaborate cryptographic scheme involving nested digital signatures.
Another concern is privacy: since the execution of the protocol requires the parties to know the topology of 
the whole swap digraph, they also receive information about asset transfers between other parties.

Herlihy's paper~\cite{Herlihy18} also mentions a simplified and more efficient version of his protocol.
In terms of cryptographic tools, this simpler protocol uses only standard smart contracts
called \emph{hash time-lock contracts} (HTLC's), that require only one secret/hashlock pair
and a time-out mechanism. This protocol, however, has only limited applicability: it works correctly only 
for special types of digraphs, that we call \emph{bottleneck digraphs} in our paper\footnote{%
Herlihy calls such graphs \emph{single-leader digraphs.}}. These are strongly connected digraphs in which all cycles
overlap on at least one vertex.

This raises a natural question, already posed in~\cite{Herlihy18}: \emph{Is there a simple and efficient protocol
for multi-party asset swaps that is atomic and works on all strongly connected digraphs?}


\smallskip

\emph{Our contributions.}
Our work was motivated by the above question. Specifically, the question we address is this:
\emph{Is it possible to solve the cross-chain swap problem
with an atomic protocol that is computationally efficient and uses only HTLCs for asset transfers?} 
We show that, while such HTLC-based protocols are not possible for arbitrary digraphs, 
the class of digraphs that have such protocols is much broader than just bottleneck digraphs.
In fact, we give a complete characterization of digraphs that have HTLC-based protocols.
We call them \emph{reuniclus graphs}\footnote{Graphical representations of
these digraphs resemble visually a type of Pokemon called Reuniclus.}.  A reuniclus graph consists of a number of (weakly)
biconnected components, with each component being an induced bottleneck subgraph. These components form a
hierarchical tree-like structure, where the bottleneck of each non-root component 
is an articulation vertex belonging to its parent component.
(See the formal definition in Section~\ref{subsec: protocol for reuniclus digraphs}).

To establish our characterization, we need to formalize the concepts of HTLCs and
HTLC-based protocols. In our model, an HTLC is a smart contract used to transfer an asset from
its seller to its buyer that (roughly) works like this: After it's created by the seller, the contract stores (escrows)
this asset secured with a hashlock value and with a timeout (expiration) value, both provided by the seller. 
To claim this asset, the buyer must provide a secret value that matches the hashlock. 
If the asset is not successfully claimed by the buyer before its timeout,
the contract automatically returns the asset to the seller. In a HTLC-based protocol, 
all asset transfers are implemented with HTLCs, and
each party is allowed to generate one secret/hashlock pair, with the hashlock value obtained from
the secret value using a one-way permutation.

\smallskip

In this terminology, our main contribution can be stated as follows:

\begin{theorem}\label{thm: main theorem}
A swap digraph $G$ has an atomic HTLC-based protocol if and only if $G$ is a reuniclus digraph.
\end{theorem}

The proof of Theorem~\ref{thm: main theorem} conists of two parts. The sufficiency condition
is proved in Section~\ref{sec: an atomic protocol for reuniclus digraphs}, 
where we provide an HTLC-based atomic protocol for  asset swaps in reuniclus digraphs. 
Our most technically challenging contribution is the proof of the necessity condition,
presented in Section~\ref{sec: characterization of digraphs that admit shl protocols}.
This proof shows that non-reuniclus digraphs do not have atomic HTLC-protocols for asset swaps.

Along the way we also obtain a full characterization of digraphs that admit HTLC-protocols
when only one party is allowed to create a secret/hashlock pair.
Herlihy~\cite{Herlihy18} showed that bottleneck graphs have such
protocols. We complement this result by proving that the bottleneck property is also necessary. 

Our asset-swap model is in fact a slight generalization of the one in~\cite{Herlihy18}, as it uses
a relaxed definition of the preference relation, which allows
each party to customize some of the preferences between its outcomes. 
As we explain in Section~\ref{sec: multi-party asset swaps}, all crucial properties of the model in~\cite{Herlihy18} remain valid. 
Among these, the most critical is given in Lemma~\ref{lem: live+safe->nash}, which says that in this model
the strong Nash equlibrium property is in essence a consequence of the safety property of a protocol;
thus, in a way, it comes ``for free''.


\smallskip
\emph{Related work.}
The problem of securely exchanging digital products between two untrustful parties
has been studied since 1990s under the name of \emph{fair exchange}. As simultaneous
exchange is not feasible in a typical electronic setting, protocols for fair exchage 
rely on a trusted party and use cryptographic tools to reduce the need for
this trusted party to interfere
-- see for example~\cite{micali03,franklin98,probfairexchange,asokan97,asokan1997}.
In the asset-swap model from~\cite{Herlihy18}, adopted in our paper, one can
think of smart contracts as playing the role of trusted parties.

Starting in early 2000's, with different users holding assets on a quickly growing number of different blockchains, 
cross-chain interoperability tools became necessary to allow these users to trade their assets.
An atomic swap concept was one of the proposed tools to address this issue. The concept itself
and some early implementations of asset-swap protocols (see, for example~\cite{2013_nolan_atomictransfertalk})
predate the work of Herlihy~\cite{Herlihy18}, although his paper was first to provide 
systematic treatment of such protocols and extend it to multiple parties.

In recent years there has been intensive research activity aimed at improving
various aspects of asset-swap protocols. The preference relation of the participants in the model from~\cite{Herlihy18}
is very rudimentary, and some
refinements of this preference model were studied in~\cite{2023_chan_etal_with_preferences,2019_imoto_etal_cross_chain_improved}.
Some proposals~\cite{mazumdar2022faster,2021_xue_herlihy_hedging} address the issue of `` griefing'', when one party needs to wait for the counter-party to
act, while its assets are locked and unaccessible. 
Other directions of study include investigations of protocol's time and space complexity~\cite{2019_imoto_etal_cross_chain_improved},
privacy issues~\cite{deshpande2020privacy}, and generalizations of swaps to more
complex transactions~\cite{herlihy2019cross,thyagarajan2021universal,heilman2020arwen}.


%% file: 02_graph_terminology.tex


Let $G = (V,A)$ be a digraph with vertex set $V$ and arc set $A$.
Throughout the paper we assume that digraphs do not have self-loops or parallel arcs.
If $(u,v)\in A$ is an arc, we refer to $v$ as an \emph{out-neighbor} of $u$ and
to $u$ as an \emph{in-neighbor} of $v$.
By $\inNeighbors{v}$ we denote the set of in-neighbors of $v$, by $\inArcs{v}$
its set of in-coming arcs,
$\outNeighbors{v}$ are the out-neighbors of $v$ and $\outArcs{v}$ are its outgoing arcs.

To simplify notation, for a vertex $v$ we will sometimes write $u\in G$ instead
of $v\in V$. Similarly, for an arc $(u,v)$, notation $(u,v)\in G$ means that
$(u,v)\in A$.

We use standard terminology for directed graphs. A path in $G$ is a sequence of
vertices, each (except last) having an arc to the next one, without repeating an arc.
A cycle is a path in which the last vertex is equal to the first one.
A \emph{simple path} is a path that does not repeat a vertex. 
A \emph{simple cycle} is a cycle where all vertices are different except the
first and last being equal.

A digraph $H = (V',A')$ is a subgraph of $G$ if $V'\subseteq V$ and $A'\subseteq A$.
$H$ is called an \emph{induced subgraph} of $G$ if $A'$ contains all arcs from $A$
connecting vertices in $V'$. The concept of induced subgraphs will be important. 

Recall that a digraph $G$ is called \emph{strongly connected} if for every pair of
vertices $u$ and $v$ there is a path from $u$ to $v$. We are mostly
interested in strongly connected digraphs.

A set $U\subseteq V$ of vertices is called a \emph{feedback vertex set} in a digraph $G$
if every cycle in $G$ includes at least one vertex from $U$. 

For a digraph $G$ we will sometimes consider its undirected structure.
A vertex $v$ is caled an \emph{articulation vertex} of $G$ if  
removing $v$ from (the undirected version of) $G$ creates a graph that has
more connected components than $G$. We say that $G$ is \emph{biconnected} 
if $G$ is (weakly) connected and has no articulation vertices.


%% file: 03_multi-party_asset_swaps.tex


In broad terms, the \emph{multi-party asset swap problem} we want to solve is this: There is a collection $V$ of parties.
Each party has a collection of assets that it wishes to exchange for some other set of assets.
For example, Alice may own a pair of skis and a pair of ski boots that she would like to exchange for
an iPad and a pair of earbuds.
(In practice, these are typically digital assets, like crypto-currency, but we can also
consider a scenario where ownership titles to physical assets are swapped.) Suppose that 
there is a way to match these exchanges perfectly, namely to reassign assets from its current owner to its 
new owner in such a way that each party will be assigned exactly their specified collection of desired assets.
This reassignment is called a \emph{multi-party asset swap}.
We now would like to design a protocol to arrange the transfers of all assets in this swap. We do not assume
that the parties are honest. Some parties may not follow the rules of the protocol, attempting to achieve an outcome
better than the one they originally specified. (For example, they may end up with an extra new asset,
or retain some of their own assets.) Other parties may just behave irrationally. 
To address this, the asset-swap protocol must satisfy, at the minimum, the following two properties:
(i) if all parties follow the protocol then all prearranged asset transfers will take place, and 
(ii) an outcome of any honest party (that follows the protocol) is guaranteed to be
acceptable (not worse than its initial holdings),
even if some other parties deviate from the protocol or collaboratively attempt to cheat.


\paragraph{Clearing service.}
We assume the existence of a \emph{market clearing service}. Each party submits its proposed
exchange (the two collections of its current and desired assets) to this clearing service. 
If a multi-party swap is possible, the clearing service 
constructs a digraph $G = (V,A)$ representing this swap. Each arc $(u,v)$ of $G$ represents 
the intended transfer of one asset from its current owner $u$ to its new owner $v$.
(In the context of this transfer, we refer to $u$ as ``the seller'' and to $v$ as ``the buyer''.)
For simplicity, we assume that any party can transfer only one asset to any other party.
This way $G$ does not have any parallel arcs. In our notations
we will also identify assets with arcs, so $(u,v)$ denotes both an arc of $G$
and the asset of $u$ to be transferred to $v$.
The market clearing service ensures that $G$ is strongly connected and
satisfies other assumptions of the swap protocol, if there are any.
It then informs each party of the steps of the swap protocol that this party should execute.
Importantly, we \emph{do not} assume that the parties trust the market clearing service.


\paragraph{Secrets and hashlocks.}
Our protocols use a cryptographic tool of one-way functions. We allow
each party $v$ to create a \emph{secret value} $s_v$,
and convert it into a \emph{hashlock value} $h_v = H(s_v)$, where $H()$ is a one-way
permutation.  The value of $s_v$ is secret, meaning that
no other party has the capability to compute $s_v$ from $h_v$. The hashlock values can be made public.

 
 \paragraph{Hash time-lock contracts (HTLCs).}
Asset transfers are realized vwith smart contracts, which in practice are simply appropriate pieces of code running on a blockchain.
For our purpose, the internal processing within smart contracts is not relevant; in our model these are
just ``black-box'' objects with a specified functionality. 
The contracts used in our model are called \emph{hash time-lock contracts}, or \emph{HTLCs}, for short,
and are defined as follows: Each contract is associated with
an arc of $G$.  For an arc $(u,v)$, its associated contract is used to transfer the
asset $(u,v)$ from $u$ to $v$. It is created by party $u$, with $u$
providing it with the asset, timeout value $\tau$, and a hashlock value $h$. Once this contract is created,
the possession of the asset is transferred from $u$ to the contract; in other words,
the asset is in this contract's escrow. The counter-party $v$ can access the contract to verify
whether its correctness; in particular, it can learn the hashlock value $h$.
There are two ways in which the asset can be released:
\begin{itemize}\setlength{\itemsep}{-0.03in}
\item
One, $v$ (the buyer) can claim it. To claim it successfully, $v$ must provide a
value $s$ such that $H(s) = h$ not later than at time $\tau$. When this happens, the smart contract 
transfers the asset to $v$, and it also gives the value of $s$ to $u$.
One can think of this as an exchange of the asset for the secret value $s$. 
\item
Two, the contract can expire. As soon as the current time exceeds $\tau$, and if the asset has not been
successfully claimed, the contract returns the asset to its seller $u$.
\end{itemize}
All contracts considered in our paper are HTLCs, so the terms ``contract'' and ``HTLC'' are from now on synonymous.
Further overloading notation and terminology, we will also refer to the contract on arc $(u,v)$
as ``contract $(u,v)$''. If this contract has hashlock $h_x$ of a party $x$ (where $x$ may be different from $u$ and $v$),
we will say that it is \emph{protected by hashlock $h_x$} or simply \emph{protected by party $x$}. 
Note that this does not mean that this contract was created by $x$.


\paragraph{Swap protocols.}
We assume that the time is discrete, consisting of unit-length time steps, starting at time $0$.
A swap protocol $\protP$ specifies for each party and each time step what actions this
party executes at this step.  These actions may involve internal computation or external actions, like communication.

In an execution of $\protP$ there is no guarantee that all parties actually follow $\protP$.
It may happen that some parties may deviate from $\protP$ for some reason,  
for example in an attempt to improve their outcomes, or they may just behave erratically.
When we refer to an \emph{honest} or \emph{conforming} party $u$, we mean that $u$ follows $\protP$,
except when it can infer from an interaction with some other party that not all parties follow $\protP$.
From that point on, $u$ may behave arbitrarily (but still rationally, so it would not do anything that
might worsen its ultimate outcome).

The execution of a protocol
results in some assets being transferred between parties. We assume that when
the process completes, even if some parties do not follow the protocol, an asset associated
with arc $(u,v)$ will end up either in the possession of $u$ or in the possession of $v$.

In an \emph{HTLC-based protocol}, all asset transfers are implemented with
HTLCs, and no other interaction between the parties is allowed. Each party is
allowed to create one secret/hashlock pair. These hashlock values are 
distributed via smart contracts. (These values can in fact be simply made public.)
A more formal definition can be found in Section~\ref{sec: characterization of digraphs that admit shl protocols}.


\paragraph{Outcomes.}
For each party $v$, \emph{$v$'s outcome} associated with an execution of a protocol $\protP$ is specified by
the sets of asets that are relinguished and acquired by $v$ in this execution.
Thus such an outcome is a pair $\outcome = \outcomepair{\outcomein{}}{\outcomeout{}}$, where
$\outcomein{} \subseteq \inArcs{v}$ and $\outcomeout{} \subseteq \outArcs{v}$. 
To reduce clutter, instead of arcs, in $\outcomepair{\outcomein{}}{\outcomeout{}}$ we
can list only the corresponding in-neighbors and out-neighbors of $v$; for example, instead of 
$\outcomepair{\braced{(x,v),(y,v)}}{\braced{v,z}}$ we will write $\outcomepair{x,y}{z}$.

An outcome $\outcome = \outcomepair{\outcomein{}}{\outcomeout{}}$ of some party $u$ is called \emph{acceptable}
if in this outcome $u$ retains all its own assets or it gains all incoming assets.
That is, either $\outcomein{} = \inArcs{v}$ or $\outcomeout{} = \emptyset$ (or both).
Following Herlihy~\cite{Herlihy18}, we define several types of outcomes for each party $v$:
\begin{description}\setlength{\itemsep}{-0.03in}
\item{$\Dealv{v} = \outcomepair{\inArcs{v}}{ \outArcs{v}}$} represents an outcome
	where all prearranged asset transfers involving $v$ are completed.
\item{$\NoDealv{v} = \outcomepair{\emptyset}{ \emptyset}$} represents an outcome
	where none of prearranged asset transfers involving $v$ is completed.
\item{$\Discountv{v} = \braced{ \outcomepair{\inArcs{v}}{\outcomeout{}} \suchthat \outcomeout{}\neq  \outArcs{v} }$}.
	That is, $\Discountv{v}$ is the set of outcomes in which all of $v$'s incoming
	asset transfers are completed, but not all outgoing transfers are.
\item{$\FreeRidev{v} = \braced{ \outcomepair{\outcomein{}}{\emptyset} \suchthat \outcomein{} \neq \emptyset}$}.
	That is, $\FreeRidev{v}$ is the set of outcomes in which  none of $v$'s outgoing
	asset transfers is completed, but some of its incoming transfers are.
\end{description}
We will skip the subscript $v$ in these notations whenever $v$ is understood from context.
These four types of outcomes are exactly all acceptable outcomes, that is $\outcome$ is acceptable if
and only if
\begin{equation*}
	\omega \;\in\; \braced{\NoDeal} \cup \braced{\Deal} \cup \FreeRide \cup \Discount
\end{equation*}
(This is in fact how~\cite{Herlihy18} defines the acceptable outcomes.)
All other outcomes are of type $\Underwater$ and are considered unacceptable.

For a set $C$ of parties, its set $\inArcs{C}$ of incoming arcs consists of
arcs $(u,v)$ with $u\notin C$ and $v\in C$. The set $\outArcs{C}$ of outgoing
arcs is defined analogously. With this,
the concept of outcomes and its properties, as defined above,
extend naturally to sets of parties (that we refer
occasionally as ``coalitions''). For example, an outcome of $C$ is
\emph{acceptable} if it either contains all incoming arcs of $C$ or
does not contain any outgoing arcs of $C$.


\paragraph{The preference relation.}
A \emph{preference relation} of a party $v$ is a partial order on the set of all outcomes for $v$ that satisfies the following three properties:

\begin{itemize}\setlength{\itemsep}{-0.05in}
\item [(\emph{p1})] If $\outcomein{1} \subseteq \outcomein{2}$ and $\outcomeout{1}\supseteq \outcomeout{2}$, 
				then $\outcome{_2}$ is preferred to $\outcome{_1}$. In other words, it is better to receive more assets and relinguish fewer assets.
\item [(\emph{p2})] If $\outcome \in \Underwater$ then $\NoDeal$ is preferable to $\outcome$. 
\item [(\emph{p3})] $\Deal$ is better than $\NoDeal$. Otherwise, $v$ would have no incentive to participate in the protocol. 
\end{itemize}

The preference relation of a party $v$ captures which outcomes are more desirable for $v$. Its intended to
capture rational behavior of parties, leading to the definition of Nash equilibrium property, given below.


\begin{figure}[t]
\begin{center}
\includegraphics[width = 4.5in]{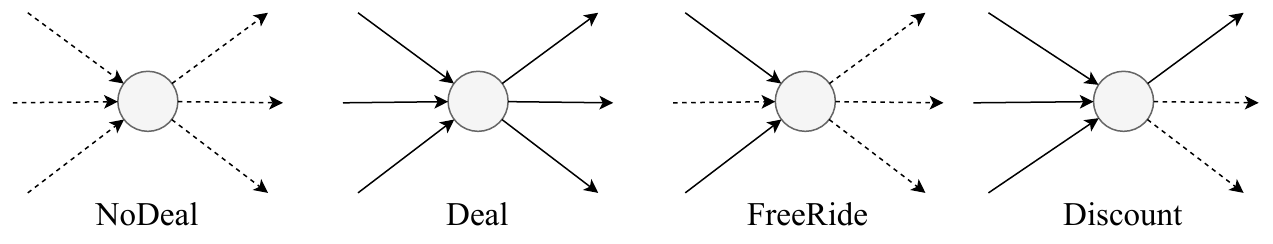}
\caption{Four types of acceptable outcomes. Solid arrows represent transferred assets and dotted arrows represent those that are not.}
\end{center}
\end{figure}


\begin{figure}[t]
\begin{center}
\includegraphics[width=5in]{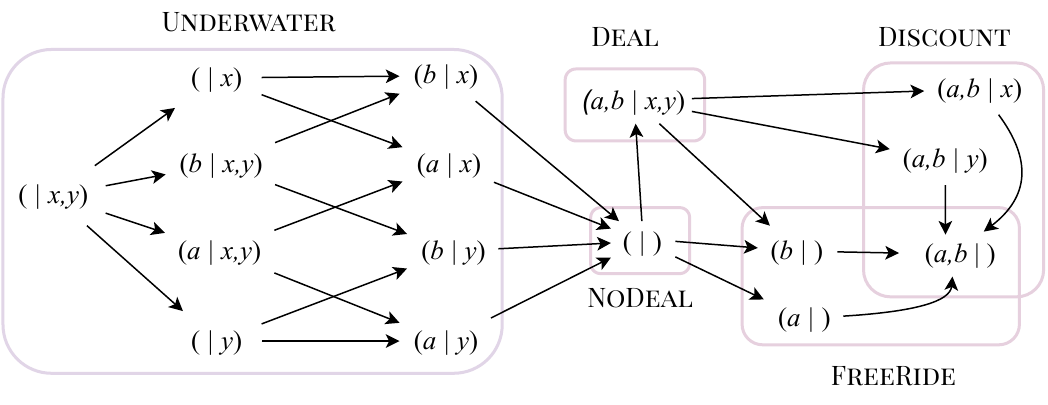}
\caption{An example of a preference relation of a node $v$ whose neighborhoods are $\inNeighbors{v} = \braced{a,b}$ and $\outNeighbors{v} = \braced{x,y}$.
Arrows represent preferences.}\label{poset}
\end{center}
\end{figure}


\paragraph{Protocol properties.}
Following~\cite{Herlihy18}, we define the following properties of a swap protocol $\protP$:
\begin{description}\setlength{\itemsep}{-0.03in}
	
	\item{\emph{Liveness}:} $\protP$ is called \emph{live} if all assets are transferred
		(that is, each party ends up in $\Deal$), providing that all parties follow $\protP$.
	
	\item{\emph{Safety}:} $\protP$ is called \emph{safe} if each honest party
		ends up in an acceptable outcome, independendently of the behavior of other parties.
		
	\item{\emph{Strong Nash Equalibrium}:}
		$\protP$ has the \emph{strong Nash equilibrium property} if
		no party has an incentive to deviate from the protocol. More precisely,
		for any set $C$ of parties (a coalition), if all parties outside $C$
		follow the protocol then the parties in $C$ cannot improve the
		outcome of $C$ by deviating from the protocol.

	\item{\emph{Atomicity}:} $\protP$ is called \emph{atomic} if its live, safe,
		and has the strong Nash equilibrium property.

\end{description}


The lemma below is a mild extension of the one given by Herlihy~\cite{Herlihy18}.
The difference is that in our definition of the preference relation we allow
some $\FreeRide$ outcomes to be preferable to $\Deal$, which was not the case in~\cite{Herlihy18}. 
The point of the lemma is that, in Herlihy's preference model, the
strong Nash equilibrium property comes for free, namely that each protocol 
has this property as long as it satisfies the liveness and safety properties.
The strong connectivity assumption is necessary for the safety property to hold, see~\cite{Herlihy18}. 

\begin{lemma}\label{lem: live+safe->nash}
Assume that digraph $G$ is strongly connected. If a protocol $\protP$ is live and safe then $\protP$ is 
atomic.
\end{lemma}

\begin{proof}
Assume that $\protP$ is live and safe. To prove that $\protP$ is atomic, we need to show that it satisfies
the strong Nash equilibrium property, as defined above. 

Towards contradiction, suppose that some coalition $C$ has an outcome $\outcome$ preferable to $\Deal$, 
even though all other parties outside $C$ follow $\protP$. Since $\outcome$ is preferable to $\Deal$, it
must be either in $\Discount$ or in $\FreeRide$. We now consider these two cases, in each
reaching a contradiction. (See Figure~\ref{fig: nash lemma cases} for illustration.)

\medskip
\noindent
\mycase{1} $\outcome\in\Discount$. That is, all incoming assets of $C$ are transferred, but
some outgoing asset $(x,y)$ of $C$ is not. Since $G$ is strongly connected, there is a 
path $Q = y_{1}y_{2}. . . y_k$ in $G$  from  $y = y_1$ to $x= y_k$. 
Let $y_{j}$ be the first vertex on this path that is in $C$. By the choice of $y_1$, we have $j > 1$.
As $y_1$ follows the protocol, the safety condition guarantees that its outcome is acceptable.
Since $(x,y_1)$ is not transferred and $y_1$'s outcome is acceptable, all outgoing assets of
$y_1$ are also not transferred.  
In particular, asset $(y_1,y_2)$ is not transferred.
The same argument gives us that if $j > 2$ then all outgoing assets of $y_2$ are not transferred.
Repeating this for all vertices $y_1,y_2,...,y_{j-1}$, we will eventually obtain that
the asset $(y_{j-1},y_{j})$ is not transferred. But $(y_{j-1},y_{j})$ is an incoming asset of $C$,
so $C$'s outcome $\outcome$ cannot be in $\Discount$.

\begin{figure}[t]
\begin{center}
\includegraphics[width = 2.2in]{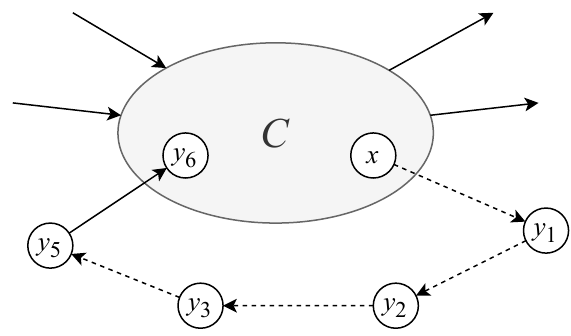}
\hspace{0.5in}
\includegraphics[width = 2.2in]{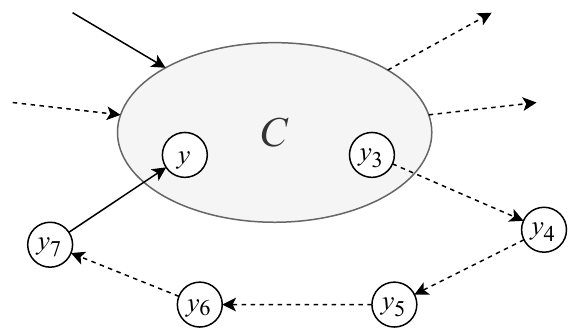}. 
\end{center}
\caption{Illustration of Cases~1 and~2 in the proof of Lemma~\ref{lem: live+safe->nash}. 
Solid arrows represent transferred assets and dotted arrows represent those that are not.
In the figure on the left, $j=6$. In the figure on the right, $k = 7$ and $j=3$.
}
\label{fig: nash lemma cases}
\end{figure}

\medskip
\noindent
\mycase{2} $\omega\in \FreeRide$. That is, at least one incoming asset of $C$ is transferred, but no outgoing asset of $C$ is transferred.
 Consider a transferred asset $(x,y)$ that is incoming into $C$. 
 Since $G$ is strongly connected, there is a path $Q = y_1 y_2\ldots y_k$ in $G$ from $y=y_1$ to $x=y_k$. 
 Let $y_j$ be the last vertex in $C$ on this path. Then $(y_j,y_{j+1})$, being an outgoing asset from $C$, is not transferred. 
 Since $y_{j+1}$ follows the protocol, the safety condition guarantees that its outcome is acceptable. 
 Since $(y_j, y_{j+1})$ is not transferred and $y_{j+1}$'s outcome is acceptable, all outgoing asset of $y_{j+1}$ are also not transferred. 
 In particular, $(y_{j+1},y_{j+2})$ is not transferred. The same argument gives us that if $j+2\le k$ then all outgoing assets of $y_{j+2}$ are not transferred.
  Repeating this for all remaining vertices on the path, we will eventually obtain that  asset $(y_{k},y) = (x,y)$ is not transferred, which is a contradiction.
So $C$'s outcome $\omega$ cannot be in $\FreeRide$.

\end{proof}

%% file: 04_atomic_protocol_for_reuniclus_digraphs.tex


In this section we give an atomic asset-swap protocol for reuniclus digraphs.
The concept of reuniclus digraphs was introduced informally in
Section~\ref{sec: introduction}, and the formal definition is given below
in Section~\ref{subsec: protocol for reuniclus digraphs}.

We present this protocol in two steps. First, in Section~\ref{subsec: protocol for bottleneck digraphs}, 
we describe an atomic protocol for bottleneck digraphs called Protocol~$\ProtBDP$ (short for \emph{bottleneck-digraph protocol}).
This protocol is essentially equivalent to a simplified version of Herlihy's protocol (see~\cite{Herlihy18}, Section 46). 
We include it here for the sake of completeness, as it  is an essential ingredient of 
our full protocol for reuniclus digraphs, called Protocol~$\ProtRDP$
(short for \emph{reuniclus-digraph protocol}), that is
presented in Section~\ref{subsec: protocol for reuniclus digraphs}.


\subsection{Protocol~$\ProtBDP$ for Bottleneck Digraphs}
\label{subsec: protocol for bottleneck digraphs}

\input{04.1_protocol_BDP_for_bottleneck_digraphs.tex}


\subsection{Protocol~$\ProtBDP$ for Reuniclus Digraphs}
\label{subsec: protocol for reuniclus digraphs}
\input{04.2_protocol_RDP_for_reuniclus_digraphs}


%% file: 04.1_protocol_BDP_for_bottleneck_digraphs.tex


A vertex $v$ in a digraph $G$ is called a \emph{bottleneck vertex} if it belongs to each cycle of $G$.
In other words, $v$ is a singleton feedback vertex set of $G$. If $G$ is strongly connected and
has a bottleneck vertex then we refer to $G$ as a \emph{bottleneck digraph}.

We now describe Protocol~$\ProtBDP$, an atomic asset-swap protocol for bottleneck digraphs.
Given a bottleneck digraph $G$, one bottleneck vertex of $G$ is designated as the \emph{leader}.
This leader, denoted $\leader$, creates its secret/hashlock pair $(s_\leader,h_\leader)$.
The other vertices are called \emph{followers}. Protocol~$\ProtBDP$ has two phases.
The first phase, initiated by $\leader$, creates all contracts. Each follower waits for all the incoming
contracts to be created, and then creates the outgoing contracts. The timeout values for 
all incoming contracts are strictly larger than the timeout values for all outgoing contracts.
In the second phase the assets are claimed, starting with $\leader$ claiming its incoming assets.
Now the process proceeds backwards. For each follower $v$, when any of its outgoing assets is claimed,
$v$ learns the secret value $s_\leader$, and it can now claim its incoming assets.

The detailed description of this protocol is given in Figure~\ref{fig: protocol BDP}. In the protocol, 
$\maxDfromLeader{v}$ denotes the \emph{maximum distance from $\leader$ to $v$}, defined as
the maximum length of a simple path from $\leader$ to $v$. In particular, $\maxDfromLeader{\leader} = 0$. 
For $y\neq\leader$, the value of $\maxDfromLeader{y}$ can be computed recursively using the formula
\begin{equation}
	\maxDfromLeader{y} \;=\; \max_{z \in\inNeighbors{y}} \maxDfromLeader{z} + 1
	\label{eqn: distance from leader}
\end{equation}
This is in essence the standard recurrence for computing the longest path in acyclic digraphs.
(One can think of $G$ as a DAG by splitting $\leader$ into two vertices, 
one with the outgoing arcs of $\leader$ and the other with the incoming arcs of $\leader$.)
The values $\maxDfromLeader{y}$ are used in contract creation times.

By $\ContractCreationTime$ we denote the maximum length of a simple cycle in $G$. This cycle is
just a path from $\leader$ to $\leader$ that does not repeat any vertices except for $\leader$,
so  $\ContractCreationTime = \max_{z \in\inNeighbors{\leader}} \maxDfromLeader{z} + 1$,
analogously to~\eqref{eqn: distance from leader}. 

In the timeout values, the notation $\maxDtoLeader{v}$ is the maximum distance from $v$ to $\leader$. 
These values satisfy a recurrence symmetric to~\eqref{eqn: distance from leader}, namely
$\maxDtoLeader{\leader} = 0$ and for $y\neq \leader$
\begin{equation}
	\maxDtoLeader{y} \;=\; \max_{z \in\outNeighbors{y}} \maxDtoLeader{z} + 1
	\label{eqn: distance to leader}
\end{equation}
Naturally, we then also have $\max_{z \in\outNeighbors{\leader}} \maxDtoLeader{z} + 1 = \ContractCreationTime$.
Note that, for each party $v$, the timeouts of all incoming contracts $(u,v)$ are equal 
$\ContractCreationTime + \maxDtoLeader{v}$, exactly the time when $v$ is scheduled to claim them.
And, by equation~\eqref{eqn: distance to leader}, if $v\neq\leader$ then
$\ContractCreationTime + \maxDtoLeader{v}$ is larger
than the timeout $\ContractCreationTime + \maxDtoLeader{w}$ of each outgoing contract $(v,w)$.

Figure~\ref{fig: protocol BDP example} shows an example of a bottleneck digraph and its timeout values
in Protocol~$\ProtBDP$.


\begin{figure}[t]
\noindent
\framebox{
\begin{minipage}[t]{2.75in}
\noindent
\textbf{Protocol} BDP for leader $\leader$:
\vspace{-.075in}
\begin{description}[leftmargin=*,itemsep=-0.04in]
\item{\textit{{\onestep} At time $0$:}} Create a secret $s_\leader$ and compute 
	$h_\leader = H(s_\leader)$. For each arc $(\leader,v)$,
	create the contract with hashlock $h_\leader$ and timeout $\tau_{\leader v} = \ContractCreationTime + \maxDtoLeader{v}$.
\item{\textit{{\onestep} At time $\ContractCreationTime$:}} Check if all the incoming contracts are
	 	created, correct, and have the same hashlock value $h_\leader$. (If not, abort.)
		Claim all incoming assets using secret $s_\leader$.
\end{description}
\end{minipage}
}
\framebox{
\begin{minipage}[t]{3.25in}
\noindent
\textbf{Protocol} BDP for a follower $u$:
\vspace{-.075in}
\begin{description}[leftmargin=*,itemsep=-0.04in]
	\item{\textit{{\onestep} At time $\maxDfromLeader{u}$:}}	Check if all the incoming contracts are created,
		 correct, and have the same hashlock value, say $h$. (If not, abort.)
		 For each arc $(u,v)$,
		 create contract with hashlock $h$ and timeout $\tau_{uv} = \ContractCreationTime + \maxDtoLeader{v}$.		 
	\item{\textit{{\onestep} At time $\ContractCreationTime + \maxDtoLeader{u}$:}} 
			Check if any of the outgoing assets was claimed. (If not, abort.)
			Let $s$ be the secret obtained from the contract for some claimed outgoing assets.
			Use $s$ to claim all incoming assets.
\end{description}
\end{minipage}
}
\caption{Protocol~$\ProtBDP$, with the protocol of the leader on the left, and the protocol of the followers on the right.
Each bullet-point step takes one time unit.
To check correctness of an incoming contract, the buyer verifies if the seller created it according to the protocol;
in particular, whether the contract contains the desired asset and whether the timeout values are correct.
}\label{fig: protocol BDP}
\end{figure}


\begin{figure}
\begin{center}
\includegraphics[width = 2.25in]{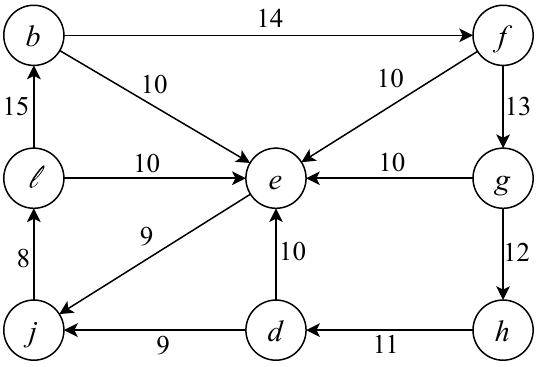}
\end{center}
\caption{An example of a bottlebeck digraph $G$ and the timeout values for Protocol $\ProtBDP$ for $G$.
The longest cycle in $G$ is $\leader,b,f,g,h,d,e,j,\leader$, so $\ContractCreationTime = 8$.}\label{fig: protocol BDP example}
\end{figure}


\begin{theorem}\label{thm: bottleneck->elementary}
If $G$ is a bottleneck digraph, then Protocol $\ProtBDP$ is an atomic swap protocol for $G$.
\end{theorem}

\begin{proof}
According to Lemma~\ref{lem: live+safe->nash}, it is sufficient to prove only the liveness and safety properties. 

\emph{Liveness.}
For the liveness property, we need to prove that if all parties follow $\ProtBDP$ then all assets will be transferred, and thus all parties end up in $\Deal$.
Assume all parties follow the protocol. In the first phase, according to  equation~\eqref{eqn: distance from leader},
each party has one time unit after all its incoming contracts are placed to create its outgoing contracts.
Thus all contracts will be successfully created and the last contract will be created by the last
in-neighbor of $\leader$, in step $\max_{z \in\outNeighbors{\leader}} \maxDtoLeader{z} = \ContractCreationTime-1$,
so before the second phase starts. 
In the second phase each asset $(u,v)$ will be claimed at time $\ContractCreationTime + \maxDtoLeader{v}$, and at that time $u$
will receive the secret $s_\leader$. By equation~\eqref{eqn: distance to leader}, $\ContractCreationTime + \maxDtoLeader{u}$
is larger than $\ContractCreationTime + \maxDtoLeader{v}$.
So $u$ will be able to sucessfully claim all incoming assets at time $\ContractCreationTime + \maxDtoLeader{u}$,
which is also the timeout value for these assets. This shows that Protocol~$\ProtBDP$ satisfies the liveness property.


\emph{Safety.}
Consider some party $v$, and assume that $v$ follows the protocol. We need to show that $v$ will end up in an acceptable outcome,
even if other parties deviate from the protocol.
In other words, we need to show that either all incoming assets of $v$ are transferred or all outgoing assets of $v$ are not transferred.
We consider two cases, when $v$ is the leader and when $v$ is a follower.

\noindent
\mycase{1} $v=\leader$. 
If some incoming contract of $\leader$ is not created before time $\ContractCreationTime$, $\leader$ will abort the protocol\footnote{%
It may be tempting, if $\leader$ has any incoming contract created, to claim its asset, even if some other incoming contract is not created.
This, however, would make the protocol unsafe, because then $\leader$ could lose some of its outgoing assets.}.
Since $\leader$ does not reveal its secret, none of $\leader$'s outgoing assets can be claimed. So all outgoing contracts
will expire and their assets will be returned to $\leader$. In this case, $\leader$’s outcome is $\NoDeal$, so it is acceptable.

Next, suppose that all incoming contracts of $\leader$ are created before time $\ContractCreationTime$.
According to the protocol, each incoming contract $(u,\leader)$ of $\leader$ has timeout 
$\tau_{u \leader} = \ContractCreationTime +\maxDtoLeader{\leader} = \ContractCreationTime$. 
Thus $\leader$ can claim all incoming assets at time $\ContractCreationTime$,
so its outcome will be $\Deal$ or $\Discount$, both acceptable. 

\noindent
\mycase{2} $v\neq \leader$. 
If some incoming contract of $v$ is not created before time $\maxDfromLeader{v}$, $v$ likewise will abort the protocol.
Then none of $v$'s outgoing assets will be claimed and 
the incoming contracts will expire, so the outcome will be $\NoDeal$, which is acceptable. 

Next suppose that at time $\maxDfromLeader{v}$ all incoming contracts of $v$ are created, correct, and all have the same
hashlock value $h$. Then $v$ will create all outgoing contracts with the same hashlock $h$.
(We remark here that, as other parties may not follow the protocol, it is not necessarily true that $h = h_\leader$.)
At time $\ContractCreationTime + \maxDtoLeader{v}$,
if none of the outgoing assets of $v$ are claimed then the outgoing assets will be
returned to $v$, so $v$'s outcome will be $\NoDeal$ or $\FreeRide$.
On the other hand, if any asset $(v,w)$ is claimed, it must be claimed no later than at time
$\tau_{vw} = \ContractCreationTime + \maxDtoLeader{w}$, and at this time $v$ will obtain a secret value $s$.
Since $v$ used the same hashlock $h$ on the outgoing contracts as the one on the incoming
contracts, it must be that $H(s) = h$, for otherwise $w$ would not be able to successfully claim $(v,w)$.
So $s$ will also work correctly for any incoming contract.
Since $\tau_{vw} < \ContractCreationTime + \maxDtoLeader{v}$, and  $\ContractCreationTime + \maxDtoLeader{v}$
is the timeout of all $v$'s incoming contracts, $v$ will successfully claim all 
incoming assets at time $\ContractCreationTime + \maxDtoLeader{v}$.
Thus $v$’s outcome would be either $\Deal$ or $\Discount$, both of which are acceptable.
\end{proof}

We remark that some optimizations can reduce the number of steps needed in Protocol~$\ProtBDP$.
For example, it is not necessary for a follower $u$ to wait until step $\ContractCreationTime + \maxDtoLeader{u}$
to claim the incoming assets. Instead, $u$ can claim all its incoming assets as soon as any of its
outgoing assets is claimed.

%% file: 04.2_protocol_RDP_for_reuniclus_digraphs.tex


\myparagraph{Reuniclus digraphs.}
Let $G$ be a strongly connected digraph. We call $G$ a \emph{reuniclus digraph} if there are 
vertices $b_1,b_2,...,b_p\in G$, induced subgraphs $G_1,G_2,...,G_p$ of $G$,
and a rooted tree $\controlgraph$ whose nodes are $b_1,b_2,...,b_{p}$, with the following properties:
\begin{description}				
	\item{(rg1)} Each digraph $G_j$ is a bottleneck subgraph, with $b_j$ being its bottleneck vertex.
				We call $G_j$ a \emph{bottleneck component} of $G$.
	\item{(rg2)} If $i\neq j$, then 
	\begin{equation*}
				G_i\cap G_j  \;=\; \begin{cases}
											\braced{b_j} 	& \textrm{if $b_i$ is the parent of $b_j$ in $\controlgraph$}
											\\
					 						\emptyset 		& \textrm{otherwise}
									\end{cases}
	\end{equation*}
\end{description}
Part~(rg2) says that $G_i$ and $G_j$ are either disjoint or only share one vertex, which is the bottleneck $b_j$ of $G_j$.
We call $G_j$ the \emph{home component} of $b_j$.
We refer to $\controlgraph$ as the \emph{control tree} of $G$. (See Figure~\ref{fig: reuniclusgraph1} for an example.)
We extend the tree terminology to relations between bottleneck components,
or between bottleneck vertices and components, in a natural fashion. For example,
if $b_i$ is the parent of $b_j$ in $\controlgraph$ then
we refer to $G_i$ as the parent component of $b_j$, to $G_j$ as the child component of $b_i$.
The same convention applies to the ancestor and descendant relations.

Intuitively, a reuniclus graph $G$ can be divided into bottleneck components. Among these, most of are pairwise disjoint.
Overlaps are allowed only between two components if one is the parent of the other in the control tree $\controlgraph$,
in which case the overlap is just a single vertex that is the bottleneck of the child component.


\begin{figure}[ht]
\begin{center}
\includegraphics[valign=m,width = 4.5in]{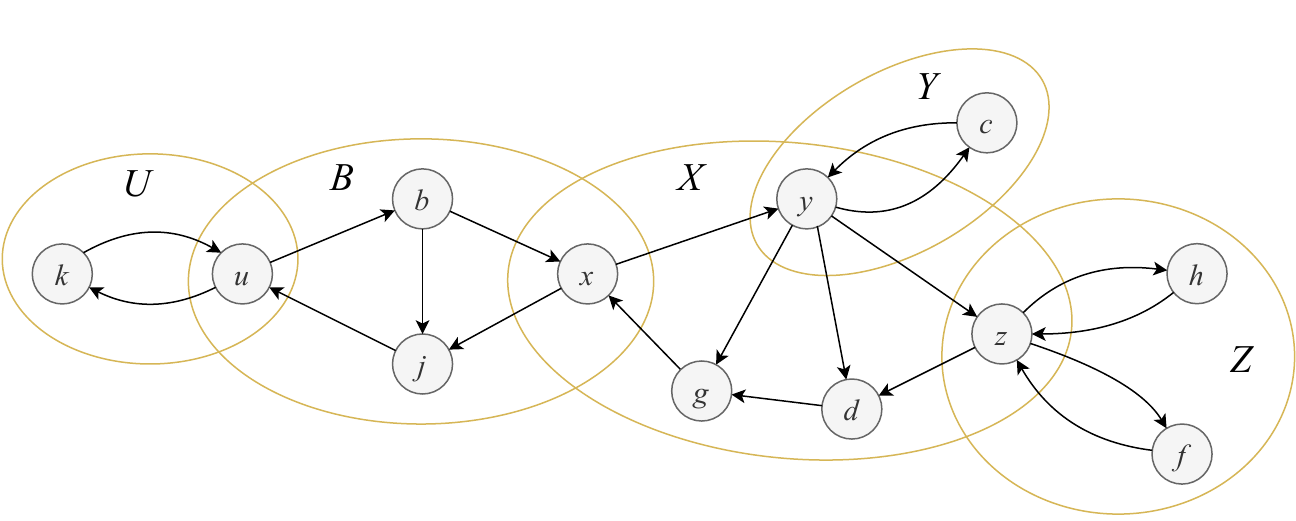}
\hspace{0.5in}
\includegraphics[valign=m,width = 1.2in]{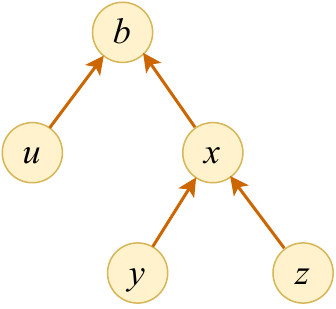}
\end{center}
\caption{An example of a reuniclus graph (left) and its control tree (right). 
The bottleneck components (circled) are $B$, $U$, $X$, $Y$ and $Z$. 
Their designated bottleneck vertices are $b$, $u$, $x$, $y$ and $z$.
Note that $Z$ consists of two biconnected components.
}\label{fig: reuniclusgraph1}
\end{figure}


From the definition, we have that the set of all bottleneck vertices in $G$ forms a feedback vertex set of $G$.
These bottleneck vertices are articulation vertices of $G$. The bottleneck components are not necessarily
biconnected; each bottleneck component may consist of several biconnected components that share the same
bottleneck vertex.


\paragraph{Protocol~$\ProtRDP$.}
Protocol~$\ProtRDP$ can be thought of as a hierachical extension of Protocol~$\ProtBDP$.
Each bottleneck vertex $b_j$ is called the leader of $G_j$. It creates its own secret/hashlock pair $(s_j,h_j)$, and
its hashlock $h_j$ is used to safely transfer assets within $G_j$, while the transfer of assets in
the descendants components of $b_j$ is ``delegated'' to the children of $b_j$ in $\controlgraph$.
We assume that the root component of $\controlgraph$ is $G_1$, and its
bottleneck $b_1$ is called the \emph{main leader}, that will be also denoted by $\leader$.
All non-leader vertices are called \emph{followers}.

Protocol~$\ProtRDP$  has two parts. In the first part
all contracts are created and in the second part the parties claim their incoming assets. 
In the first phase, at time $0$ all leaders create the outgoing contracts within their home bottleneck components.
Then the contracts are propagated within the bottleneck components, to some degree independently;
except that each leader $b_j$ creates its outgoing contracts in its parent component $G_i$ only
after all its incoming contracts, \emph{both in $G_i$ and $G_j$}, are created.
This in fact ensures that at that time all contracts in its descendant components will be already created.

In the asset claiming phase, the main leader $\leader$ is the first to claim the incoming contracts.
The behavior of followers is the same as in Protocol~$\ProtBDP$: they claim the incoming assets
one step after all their outgoing assets were claimed. 
The behavior of all non-main leaders is more subtle. 
Each such leader $b_j$ waits until all its outgoing assets \emph{in the parent component} 
are claimed, and then it claims all of its incoming assets. (So $b_j$ does not wait for the
outgoing assets in its home component to be claimed.)

The full protocol for non-main leaders $b_j$ is given in Figure~\ref{fig: protocl SHCP}. 
Figure~\ref{fig: reuniclusgraph1 timeouts} shows timeout values for the reuniclus graph in Figure~\ref{fig: reuniclusgraph1}.
In what follows we explain some notations used in the protocol.

As before, we use notation $\maxDtoLeader{y}$ for the maximum distance from $y$ to $\leader$ in $G$
(see the formula~\eqref{eqn: distance from leader} in Section~\ref{subsec: protocol for bottleneck digraphs}).
The formula~\eqref{eqn: distance to leader} remains valid as is.
We also need the concept analogous to the maximum distance from a leader, but this one is a little
more subtle than for bottleneck graphs, because we now need to consider paths whose initial bottleneck
vertex can be repeated once later on the path. Formally,
if $v\in G_i$, then $\maxDfromSubleader{v}$ denotes the maximum length of a path with the following properties:
\begin{description}\renewcommand{\itemsep}{-0.02in}
	\item{(i)} it starts at some leader $b_j$ that is a descendant of $b_i$ (possibly $b_j = b_i$), 
	\item{(ii)} it ends in $v$, and 
	\item{(iii)} it	does not repeat any vertices, with one possible exception: it can only revisit $b_j$,
		and if it does, it either ends or leaves $G_j$ (and continues in the parent component of $b_j$).
\end{description}
(This can be interpreted as a maximum path length in a DAG obtained by splitting each leader into
two vertices, one with the outgoing arcs into its home component and the other with all other arcs.)
For example, in the example in Figure~\ref{fig: reuniclusgraph1}, one allowed path 
for $v=u$ is $x-y-z-d-x-j-u$.

These values can be computed using auxiliary values $\maxDfromSubleader{uv}$ defined for each edge $(u,v)$.
Call an edge $(u,v)$ a \emph{bottleneck edge} if $u$ is a bottleneck vertex, say $u = b_j$, and $v\in G_j$.
That is, bottleneck edges are the edges from bottleneck vertices that go into their home components.
First, for each bottleneck  edge $(b_j,v)$ let $\maxDfromSubleader{b_jv} = 0$. Then, for each
vertex $u$ and each non-bottleneck edge $(u,v)$ let
\begin{equation}
			\maxDfromSubleader{u} \;=\; \maxDfromSubleader{uv} \;=\; \max_{(x,u)} \maxDfromSubleader{xu} + 1
			\label{eqn: max d from subleader}
\end{equation}
where the maximum is over all edges $(x,u)$ entering $u$. By $\RContractCreationTime$ we denote
the value of $\maxDfromSubleader{\leader}$.

The values $\maxDfromSubleader{z}$ determine contract creation times.
As shown in Figure~\ref{fig: protocl SHCP}, each leader $b_j$ will create its contracts in its home
component at time $0$. Each other contract $(u,v)$ will be created at time $\maxDfromSubleader{u}$.
The last contract will be created by some in-neighbor of $\leader$ at time step  $\RContractCreationTime-1$.
Then $\leader$ will initiate the contract claiming phase at time $\RContractCreationTime$.
Analogous to Protocol~$\ProtBDP$, each party $u$ will claim 
its incoming contracts at time $\RContractCreationTime + \maxDtoLeader{u}$, which
is its timeout value.


\begin{figure}[t]
\begin{center}
\noindent
\framebox{
\begin{minipage}[t]{5.25in}
\noindent
\textbf{Protocol} $\ProtBDP$ for a leader $b_j \in G_i\cap G_j$: 
\begin{description}[leftmargin=*,itemsep=-0.04in]
\item{\textit{{\onestep} At time $0$:}} Generate secret $s_j$ and compute $h_j = H(s_j)$. 
		For each arc $(b_j,v)$ in $G_j$, 
		create contract  with hashlock $h_j$ and timeout $\tau_{b_jv} = \RContractCreationTime + \maxDtoLeader{v}$.
\item{\textit{{\onestep} At time $\maxDfromSubleader{u}$:}}
Check if all incoming contracts are in place, if they are correct, if
all the hashlocks in the incoming contracts in 
$G_i$ have the same value $h$, and if all the hashlocks in the incoming contracts in $G_j$ are equal $h_j$. (If not, abort.)
For each arc $(b_j,v)$ in $G_i$
create its contract with hashlock $h$ and timeout $\tau_{b_iv} = \RContractCreationTime + \maxDtoLeader{v}$.
\item{\textit{{\onestep} At time $\RContractCreationTime + \maxDtoLeader{u}$:}} 
			Check if any of the outgoing assets in $G_i$ was claimed. (If not, abort.)
			Let $s$ be the secret obtained from the contract for some claimed outgoing assets in $G_i$.
			Claim all its incoming assets,
			using secret $s$ in $G_i$ and using secret $s_j$ in $G_j$.
\end{description}
\end{minipage}
}
\hspace{0.02in}
\caption{Protocol~$\ProtRDP$ for a  sub-leader $b_j$, namely the bottleneck vertex of $G_j$ that also
	 		belongs to its parent graph $G_i$. 
			Recall that $\maxDfromSubleader{u}$ is the maximum distance from some leader to $u$ along a path
			that satisfies conditions~(i)-(iii), and that $\RContractCreationTime = \maxDfromSubleader{\leader}$.
			$\maxDtoLeader{v}$ is the maximum length of a simple path from $v$ to $\leader$, 
			as defined in Section~\ref{subsec: protocol for bottleneck digraphs}.
			As explained in the text, the protocols for the main leader and the pure 
			followers are the same as in Protocol~$\ProtBDP$.}\label{fig: protocl SHCP}
\end{center}
\end{figure}



\begin{figure}
\begin{center}
\includegraphics[width = 5.25in]{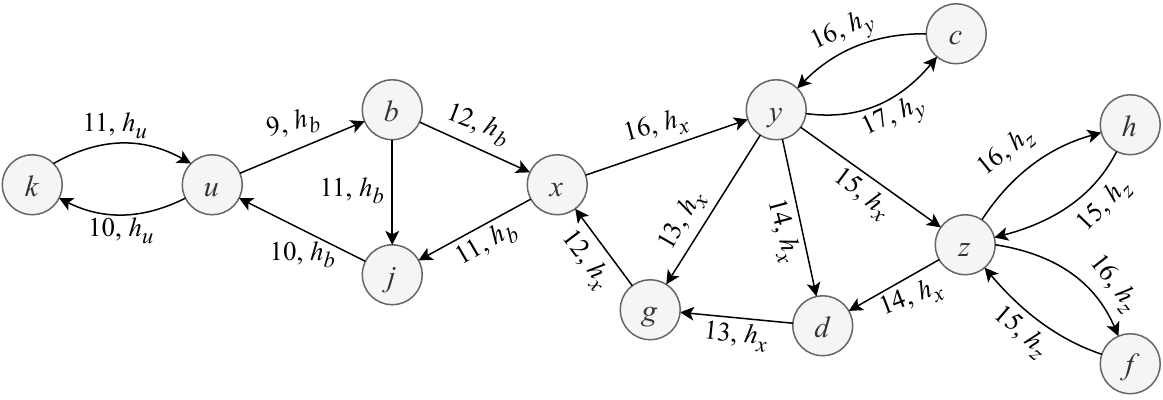}
\end{center}
\caption{Timeout values and hashlocks for Protocol~$\ProtRDP$ for the graph in Figure~\ref{fig: reuniclusgraph1}.
	The main leader is $\leader = b$. 
	We have $\RContractCreationTime = 9$ (this is the length of path $y-c-y-z-d-g-x-j-u-b$).
	}
\label{fig: reuniclusgraph1 timeouts}
\end{figure}


\begin{theorem}
If $G$ is a reuniclus digraph, then Protocol~$\ProtRDP$ is an atomic swap protocol for $G$.
\end{theorem}

\begin{proof}
According to Lemma~\ref{lem: live+safe->nash}, it is sufficient to prove only the liveness and safety properties. 


\emph{Liveness.}
The liveness property is quite straightforward. Each party $u\neq\ell$ has exactly one time unit, after its last incoming contract
is created, to create its outgoing contracts, according to~(\ref{eqn: max d from subleader}).
This will complete the contract creation at time $\RContractCreationTime-1$.
Thus at time $\RContractCreationTime$ leader $\leader$ can claim its incoming assets.
For any other party $u$, each incoming asset $(x,u)$ of $u$ has timeout $\tau_{xu} = \RContractCreationTime + \maxDtoLeader{u}$.
By formula~\eqref{eqn: distance to leader}, 
if $u$ is a follower then all the outgoing assets of $u$ will be claimed before time $\tau_{xu}$,
and if $u$ is a non-main leader then all of $u$'s ougoing assets in its parent component will be claimed before time $\tau_{xu}$.
So $u$ can claim all incoming assets at time $\tau_{xu}$.


\emph{Safety.}
The proof of the safety condition for the main leader $\leader$ and pure followers is the same as in
Protocol~$\ProtBDP$ for bottleneck digraphs. So here we focus only on non-main leaders.

Let $b_j$ be a non-main leader whose parent component is $G_i$. Assume that $b_j$ follows the protocol.
So, according to Protocol~$\ProtRDP$,
$b_j$ will create its outgoing contracts in $G_j$ at time $0$.
Before creating its outgoing contracts in $G_i$, $b_j$ checks if \emph{all} incoming contracts are
created. If any of its incoming contracts is not created or if any value in this contract
is not as specified in the protocol, $b_j$ will abort without creating its outgoing contracts
in $G_i$. Thus its outcome will be $\NoDeal$. 

So assume that all incoming contracts of $b_j$ are created and correct; in particular all incoming
contracts in $G_j$ have hashlock $h_j$ and all incoming contracts in $G_i$ have the same hashlock $h$ (which may or may not be equal to $h_i$).
Then $b_j$ creates
its outgoing contracts in $G_i$, as in the protocol. We now need to argue that if any of $b_j$'s
outgoing assets is successfully claimed then $b_j$ successfully claims \emph{all} its
incoming assets.

Suppose that some outgoing asset of $b_j$, say $(b_j,w)$ is successfully claimed by $w$.
Two cases arise, depending on whether $w$ is in $G_i$ or $G_j$.

If $w\in G_i$ (the parent component of $b_j$) then from contract $(b_j,w)$ will provide
$b_j$ with some secret value $s$ for which $H(s) = h$, because $b_j$ used $h$ for its
outgoing contracts in $G_i$. At this point, $b_j$ has both secret values $s$ and $s_j$, and, by formula~\eqref{eqn: distance to leader},
the timeout of all incoming contracts of $b_j$ is greater than the timeout of $(b_j,w)$.
Therefore $b_j$ has the correct secrets and
at least one time unit to claim all incoming contracts, and its outcome will be $\Deal$ or $\Discount$,
thus acceptable.

In the second case, $w\in G_j$, the home component of $b_j$. For $w$ to successfully claim
$(b_j,w)$, it must have the value of $s_j$. But, as $b_j$ follows the protocol,
it releases $s_j$ only when claiming all incoming assets. So at this time $b_j$ already has all incoming assets.
Therefore in this case the outcome of $b_j$ is also either $\Deal$ or $\Discount$.
\end{proof}

%% file: 05_characterization_of_digraphs_that_admit_shl_protocols.tex


In this section we provide a full characterization of digraphs that have HTLC-based protocols,
showing that these are exactly the reuniclus digraphs. The concept of HTLC-based protocols
was informally introduced in Section~\ref{sec: introduction}. For our characterization we need a more
precise definition that we provide next.


\smallskip

\myparagraph{Formal definition of HTLC-based protocols.}
We will say that a protocol $\protP$ is an \emph{HTLC-based protocol} if $\protP$
consists of discrete time steps, each taking one time unit, and in each step any party $v$
is allowed to execute any number of the following actions:
\begin{description}\setlength{\itemsep}{-0.03in}
\item{\emph{Create a secret/hashlock pair}:} $v$ can create a secret value $s_v$ and compute
	the corresponding hashlock $h_v = H(s_v)$. We assume that $v$ can create only one hashlock.
\item{\emph{Create contract}:} 
If $(v,w)\in G$, $v$ can create an HTLC for transferring the asset $(v,w)$ to $w$.
	This contract can be created only once. 
	It must be secured with a hashlock $h_x$ of some party $x$ that created a secret/hashlock pair. 
\item{\emph{Claim asset:}}
 If $(u,v)\in G$ and the contract $(u,v)$ is placed and and not yet expired (that is, the
 current time does not exceed the timeout value $\tau_{uv}$ of this contract), 
 then $v$ can claim asset $(u,v)$, providing that $v$ knows the secret value $s_y$ 
 corresponding to the hashlock value $h_y$ that protects this contract.
\end{description}
Parties have read access to the incoming contracts, so that they can verify their
correctness and obtain their hashlock values. This allows a protocol to propagate
the hashlock values. But the model also allows 
the parties to make the hashlock values public or distribute them selectively to other parties.
There are no restrictions on local computation, with one obvious exception, namely that the parties
cannot break any cryptographic tools used in smart contracts.

\smallskip

We now turn to our characterization, captured by Theorem~\ref{thm: main theorem} in Section~\ref{sec: introduction},
that we restate here for convenience:


\begin{theoremrestated}{\ref{thm: main theorem}}
A swap digraph $G$ has an atomic HTLC-based protocol if and only if $G$ is a reuniclus digraph.
\end{theoremrestated}

The rest of this section is devoted to the proof of Theorem~\ref{thm: main theorem}.
By straightforward inspection, Protocol~$\ProtRDP$ from Section~\ref{subsec: protocol for reuniclus digraphs}
is HTLC-based: each party creates at most one secret/hashlock pair, and all
contracts are created and claimed following the rules detailed above. 
This already proves 
the $(\Leftarrow)$ implication in Theorem~\ref{thm: main theorem}.

It remains to prove the $(\Rightarrow)$ implication, namely that the existence of an atomic HTLC-based protocol
implies the reuniclus property of the underlying graph. We divide the proof into
two parts. First, in Section~\ref{subsec: basic properties of shl-based protocols} we
establish some basic properties of HTLC-based protocols.
Using these properties, we then wrap up the proof of the $(\Rightarrow)$ implication
in Section~\ref{subsec: proof of characterization}.


\subsection{Basic Properties of HTLC-Based Protocols}
\label{subsec: basic properties of shl-based protocols}
\input{05.1_basic_properties_of_shl-based_protocols.tex}


\subsection{Proof of the $\boldsymbol{(\Leftarrow)}$ Implication in Theorem~\ref{thm: main theorem}}
\label{subsec: proof of characterization}
\input{05.2_proof_of_characterization.tex}


%% file: 05.1_basic_properties_of_shl-based_protocols.tex


Let $\protP$ be an HTLC-based protocol for a strongly connected
digraph $G$, and for the rest of this section assume that $\protP$ is atomic. 
We now establish some fundamental properties that must be satisfied by $\protP$. 
To be more precise, whenever we say that $\protP$ has a certain property, we mean that this property is
satisfied in the \emph{conforming execution} of $\protP$ on $G$, 
that is when all parties are \emph{conforming} (or, equivalently, \emph{honest}), that is they follow $\protP$.

Withouth loss of generality
we can assume that in $\protP$ each asset $(u,v)$ is claimed by $v$ exactly at the expiration time 
$\tau_{uv}$ of contract $(u,v)$. This is because otherwise we can decrease the
timeout value $\tau_{uv}$ to the time step when $v$ claims asset $(u,v)$ in $\protP$, and
after this change $\protP$ remains HTLC-based and atomic.
Note also that the liveness property of $\protP$ implies that $\tau_{uv}$ must be
larger than the creation time of contract $(u,v)$.

Most of the proofs of protocol properties given below use the same fundamental approach, 
based on an argument by contradiction: we show that if $\protP$ did not satisfy the given property then
there would exist a (non-conforming) execution of $\protP$ in which some parties, by deviating from $\protP$,
would force a final outcome of some conforming party to be unacceptable, thus violating the safety property.


\begin{lemma}\label{lem: all contracts placed before claiming}
If some party successfully claims an incoming asset at some time $t$,
then all contracts in the whole graph must be placed before time $t$.
\end{lemma}

\begin{proof}
Assume that a party $v$ successfully claims an asset $(u,v)$ at time $t$. 
Towards contradiction, suppose that there is some
arc $(x,y)$ for which the contract is still not placed at time $t$.
Since $G$ is strongly connected, there is a path  $y= u_1,u_2,...,u_p = u$  from $y$ to $u$ in $G$.
Let also $u_0 =x$ and $u_{p+1} = v$. Now consider an execution of $\protP$ in which all parties except $x$ are conforming,
$x$ follows $\protP$ up to time $t-1$, but later it never creates contract $(x,y)$. 
This execution is indistinguishable from the conforming execution up until time $t-1$, so
at time $t$ node $v$ will claim contract $(u,v)$.
Since the first asset on path $u_0,u_1,...,u_{p+1}$ is not transferred and the last one is, 
there will be a party $u_j$ on this path, with $1\le j \le p$,
for which asset $(u_{j-1},u_j)$ is not transferred but $(u_j,u_{j+1})$ is.
But then the outcome of $u_j$ is unacceptable even though $u_j$ is honest, 
contradicting the safety property of $\protP$.
\end{proof}

Lemma~\ref{lem: all contracts placed before claiming} is important: it implies that
$\protP$ must consist of two phases:
the \emph{contract creation} phase, in which all parties place their outgoing contracts (by the liveness property,
all contracts must be created),
followed by the \emph{asset claiming} phase, when the parties claim their incoming assets. 

Recall that by a \emph{contract protected by $x$} we mean a created contract with hashlock $h_x$. (Here,
$x$ is not necessarily the party that created the contract.)


\begin{lemma}\label{lem: creating contracts}
Suppose that at a time $t$ a party $v$ creates an outgoing contract  
protected by a party different than $v$.
Then  all $v$'s incoming contracts must be created before time $t$.
\end{lemma}

\begin{proof}
Suppose the property in the lemma does not hold for some party $v$, that is, at some time $t$, $v$ creates
a contract $(v,w)$ protected by hashlock $h_x$ of some party $x\neq v$,
while some incoming contract $(u,v)$ is not yet created.

We consider another execution of the protocol. In this execution,
$v$ is conforming. Also, up until time $t-1$ all parties follow the
protocol, so that $v$'s behavior at time $t$ will not change and it will create contract $(v,w)$.
Then at time $t$ some parties will deviate from the protocol. Specifically, 
(i) at time $t$ party $x$ will give its 
secret value $s_x$ to $w$ (this is needed only if $x\neq w$), 
(ii) using $s_x$, party $w$ will claim asset $(v,w)$ at time $t+1$,
and (iii)  $u$ will never create contract $(u,v)$.
As a result, $v$ will lose its asset $(v,w)$ without getting asset $(u,v)$, so $v$'s
outcome is not acceptable. This contradicts the safety property of $\protP$,
so we can conclude that $\protP$ must satisfy the lemma.
\end{proof}


Consider now the snapshot of of $\protP$ right after the contract creation phase, when all contracts are 
already in place but none of the assets are yet claimed. 
The corollary below follows directly from Lemma~\ref{lem: creating contracts}.

\begin{corollary}\label{cor: paths and cycles creation times}
(a)
If on some path each contract except possibly first is not protected by its seller, then along this path
the contract creation times strictly increase.
(b) 
Each cycle must contain a contract protected by its seller.
\end{corollary}


Next, we establish some local properties of $\protP$. Namely, in the next four lemmas we will show that
for each party $v$ there is at most one other party that protects contracts involving $v$.
We will also establish some relations between the timeout values of the contracts involving $v$.

\begin{lemma}\label{lem: v with incoming contract protected by x}
If a party $v$ has an incoming contract protected by some party $x$ different from $v$ then
\begin{description}\renewcommand{\itemsep}{-0.05in}
\item{(a)} Party $v$ has at least one outgoing contract protected by $x$.
\item{(b)} All contracts involving $v$ are protected either by $v$ or by $x$. 
\end{description}
\end{lemma}

\begin{proof}
Let $(u,v)$ be an incoming contract protected by $x$, where $x\neq v$.

(a) 
This claim is easy: if $v$ didn't have an outgoing contract protected by $x$, then even in the conforming
execution it would never receive the secret value $s_x$ needed to claim $(u,v)$.
That would violate the liveness property of $\protP$.

(b) 
We first show that all outgoing contracts must be protected by $v$ or $x$.
Suppose that it is not true, namely that some outgoing contract $(v,w)$ is protected by some party $y$, where $y\notin\braced{v,x}$.
Let $t$ be the time right after all contracts are created. We change the behavior
of some parties, as follows. Up to time $t-1$ they follow the protocol. At time $t$,
$y$ can provide $w$ with secret $s_y$, allowing $w$ to claim $(v,w)$ at time $t+1$.
All parties other than $v$ and $w$ don't do anything starting at time $t$. Then $v$ will not receive the secret $s_x$
needed to claim asset $(u,v)$, so it will end up
in an unacceptable outcome, contradicting the safety property.
We can thus conclude that all outgoing contracts not protected by $v$ must be protected by $x$. 

Next, consider incoming contracts not protected by $v$. If $(u',v)$ is an incoming contract of
$v$ protected by some party $x'\notin\braced{v,x}$ then, by part~(a), $v$ would have at least one
outgoing contract protected by $x'$, contradicting that all outgoing contracts are protected either by $v$ or $x$.
This completes the proof of~(b).
\end{proof}


\begin{lemma}\label{lem: path contracts not protected by seller}
Let $P = u_1,u_2,...,u_{k}$ be a simple path whose last contract is protected by some 
party $z \notin \braced{u_1,u_2,...,u_{k-1}}$.
Then for
each $i = 1,...,k-1$, contract $(u_i,u_{i+1})$ is protected by one of the parties
		$u_{i+1}, u_{i+2}, ..., u_{k-1},z$.
Consequently, each contract on $P$ is not protected by its seller.
\end{lemma}

\begin{proof}
The second part of the lemma follows trivially from the first.

We prove the first part by induction, proceeding backwards on $P$.
The claim is true for $i=k-1$, because, by our assumption, contract $(u_{k-1},u_k)$ is protected by $z$. 
Inductively (going backwards on $P$), assume that the lemma holds 
for $i = k-1,k-2,...,j$, where $j\ge 2$. 
In particular, by the inductive assumption, $(u_{j},u_{j+1})$ is protected by 
some party $u' \in \braced{ u_{j+1}, u_{j+2},...,u_{k-1}, z }$,
Since $P$ is a simple path, we have $u'\neq u_{j}$. Then
Lemma~\ref{lem: v with incoming contract protected by x} implies that
contract $(u_{j-1},u_j)$ is protected either by $u_j$ or by $u'$.
Therefore contract $(u_{j-1},u_j)$ is protected by 
one of the parties $u_{j}, u_{j+1},...,u_{k-1},z$.
So the lemma holds for $i = j-1$, completing the inductive step.
\end{proof}


\begin{lemma}\label{lem: v if all incoming  then all outgoing v}
If all incoming contracts of a party $v$ are protected by $v$ then
all outgoing contracts of $v$ are also protected by $v$.
\end{lemma}

Intuitively, if $v$ had an outgoing contract protected by some other party $x$ but
not an incoming contract protected by $x$, then this outgoing contract would
be ``redundant'' for $v$, since $v$ does not need
the secret from this contract to claim an incoming contracts. The lemma
shows that the issue is not just redundancy --- this is in fact not even possible if the protocol is atomic.

\begin{proof}
We prove the lemma by contradiction.
Suppose that all $v$'s incoming contracts are protected by $v$, but $v$
has an outgoing contract $(v,w)$ protected by some other party $x$.
Let $P = u_1,u_2,...,u_p$ be a simple path from $w$ to $v$, that is $u_1 = w$ and $u_p = v$. 
Then Lemma~\ref{lem: path contracts not protected by seller} implies that
all contracts on $P$ are not protected by their sellers.
The contract $(u_p,u_1) = (v,w)$ is not protected by its seller $v$. 
This would give us a cycle, namely $C = u_p,u_1,...,u_{p-1},u_p$,
in which each contract is not protected by the seller.
But this contradicts Corollary~\ref{cor: paths and cycles creation times}(b),
completing the proof of the lemma.
\end{proof}


\begin{lemma}\label{lem: v outgoing x => incoming x}
If a party $v$ has an outgoing contract protected by some party $x$ different from $v$ 
then it has an incoming contract protected by $x$.
\end{lemma}

\begin{proof}
This follows quite easily from the lemmas above.
Consider the incoming contracts of $v$. Not all of them can be protected by $v$, because
otherwise Lemma~\ref{lem: v if all incoming  then all outgoing v} would imply that
all outgoing contracts would also be protected by $v$, contradicting the assumption of the lemma.
So some incoming contract is protected by some party $x'\neq v$.
But then Lemma~\ref{lem: v with incoming contract protected by x}(b) implies
that all outgoing contracts must be protected by either $v$ or $x'$, so we must have $x' = x$.
\end{proof}


\begin{lemma}\label{lem: v incoming v iff outgoing v}
A party $v$ has an incoming contract protected by $v$ if and only if it has
an outgoing contract protected by $v$.
\end{lemma}

\begin{proof}
$(\Rightarrow)$ 
Suppose that $v$ has an incoming contract $(u,v)$ protected by $v$. 
By Lemma~\ref{lem: v outgoing x => incoming x}, $u$ has an incoming contract $(u',u)$ protected by $v$.
If $u'\neq v$, by the same argument, $u'$ must also have an incoming contract protected by $v$.
Repeating this process, we obtain a backward path with all contracts protected by $v$. If this path did not reach
$v$ it would have to create a cycle in which all contracts are protected by a party (namely $v$) outside this cycle, 
but this would contradict Corollary~\ref{cor: paths and cycles creation times}(b).

$(\Leftarrow)$ 
The argument is symmetric to that in part~(a).
Suppose that $v$ has an outgoing contract protected by $v$, say $(v,w)$. If $w\neq v$ then,
by Lemma~\ref{lem: v with incoming contract protected by x}, $w$ has an outgoing contract protected by $v$.
Repeating this process, we obtain a path with all contracts protected by $v$. If this path did not reach
$v$ it would have to create a cycle, but this would contradict Corollary~\ref{cor: paths and cycles creation times}(b).
\end{proof}


The theorem below summarizes the local properties of the contracts involving a party $v$.
See also the illustration of this theorem in Figure~\ref{fig: timeouts involving v}.

\begin{theorem}\label{thm: local properties}
Consider the contracts involving a party $v$, both incoming and outgoing. 
{
\begin{description}\setlength{\itemsep}{-0.05in}
\item{(a)} For each party $x$ (which may or may not be $v$), $v$ has an incoming contract protected by $x$ if and only if
			$v$ has an outgoing contract protected by $x$.
\item{(b)} If there are any contracts protected by $v$, then
		  at least one incoming contract protected by $v$ has a smaller timeout than all outgoing contracts protected by $v$.

\item{(c)} There is at most one party $x\neq v$ that protects a contract involving $v$.
		For this $x$, 
		all timeouts of the outgoing contracts protected by $x$ are smaller than
		all timeouts of	the incoming contracts (no matter what party protects them). 
\end{description}
}
\end{theorem}


\begin{figure}[ht]
\begin{center}
\includegraphics[valign=m,width = 1.8in]{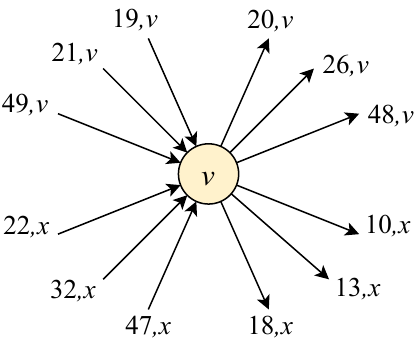}
\end{center}
\caption{Illustration of Theorem~\ref{thm: local properties}.
Each contract is labelled by its timeout and the party that protects it.
Per part~(b), the minium timout $19$ of incoming contracts protected by $v$
which is itself smaller than the minimum timeout $20$ of outgoing contracts protected by $v$.
Per part~(c), the maximum timeout $18$ of outgoing contracts protected by $x$ is smaller than
the minium timout $19$ of all incoming contracts.
}\label{fig: timeouts involving v}
\end{figure}


\begin{proof}
(a) 
This part is just a restatement of the properties established earlier.
For $x=v$, the statement is the same as in Lemma~\ref{lem: v incoming v iff outgoing v}.
For $x\neq v$, if $v$ has an incoming contract protected by $x$ then, by Lemma~\ref{lem: v with incoming contract protected by x},
it must have an outgoing contract protected by $x$, and
if $v$ has an outgoing contract protected by $x$ then, by Lemma~\ref{lem: v outgoing x => incoming x},
it must have an incoming contract protected by $x$.

(b)
Let $(v,w)$ be an outgoing contract protected by $v$ 
whose timeout value $\tau_{vw}$ is smallest.
Consider any path $P$ from $w$ to $v$ with all contracts on $P$ protected by $v$.
(This path must exist. To see this, starting from $w$ follow contracts protected by $v$.
By Corollary~\ref{cor: paths and cycles creation times}(b), eventually this process must end at $v$.)
Then part~(b) of Theorem~\ref{thm: local properties} implies that along this path
timeut values must decrease, so its last contract $(u,v)$ must satisfy $\tau_{uv} < \tau_{vw}$.
So the timeout value of $(u,v)$ is smaller than that of $(v,w)$, and thus also smaller
than all timeout values of the outgoing contracts protected by $v$.

(c)
Let $x\neq v$. If $v$ has an incoming contract protected by $x$ then, by
Lemma~\ref{lem: v with incoming contract protected by x}, all contracts involving $v$ but
not protected by $v$ are protected by $x$. If $v$ has an outgoing contract protected by $x$
then Lemma~\ref{lem: v outgoing x => incoming x} implies that some
incoming contract is protected by $x$, which leads to the same conclusion.

We now consider the claims about the timeout values.
Let $(v,w)$ be an outgoing contract protected by $x$, and let $(u,v)$ be an incoming contract.
Towards contradiction, suppose that in $\protP$ the timeouts of these contracts satisfy $\tau_{uv} \le \tau_{vw}$.
Denote by $t$ the first step of $\protP$ after the contract creation phase.
We consider two cases, depending on whether $(u,v)$ is protected by $x$ or $v$.

\mycase{1} contract $(u,v)$ is protected by $x$. 
Suppose that in $\protP$ the timeout of contract $(u,v)$ is $\tau_{uv} \le \tau_{vw}$.
We consider an execution of $\protP$ where all parties follow $\protP$ until itme $t-1$.
Then, starting at time $t$, we alter the behavior of some parties, as follows: 
all parties other than $v$, $w$ and $x$ will abort the protocol, $x$ will provide its secret $s_x$ to $w$,
and $w$ will claim asset $(v,w)$ at time $\tau_{vw}$. This way,
the earliest $v$ can claim asset $(u,v)$ is at time $\tau_{vw}+1$, which is after 
timeout $\tau_{uv}$ of $(u,v)$. Thus $v$ ends up in an unacceptable outcome,
giving us a contradiction, which completes the proof of~(b).

\mycase{2} contract $(u,v)$ is protected by $v$. We can assume that
its timeout $\tau_{uv}$ is minimum among all incoming contracts protected by $v$.
(Otherwise, in the argument below replace $(u,v)$ by the incoming contract protected by $v$ that has minimum timeout.)
Let $(v,y)$ be any outgoing contract protected by $v$. From part~(b), we have that $\tau_{uv} < \tau_{vy}$.
Let also $(z,v)$ be any incoming contract protected by $x$. 

We now consider an execution of $\protP$ where all parties follow the protocol
until time $t-1$. At time $t$, all parties other than $u,v,x,y,z$ abort the protocol,
and $x$ gives its secret $s_x$ to $w$. As time proceeds, $v$ may notice that
some parties do not follow the protocol, so, even though $v$ is honest, from this time
on it is not required to follow the protocol. We show that independently of
$v$'s behavior, it will end up in an unacceptable outcome, contradicting the safety property of $\protP$.

To this end, we consider two possibilities. If $v$ does not claim $(u,v)$ at or before time $\tau_{uv}$, 
then $w$ can claim $(v,w)$ at time $\tau_{vw}$, so $v$ will lose asset $(v,w)$ without
getting asset $(u,v)$. 
On the other hand, if $v$ claims $(u,v)$, then $u$ can give secret $s_v$ to $y$
that can then claim asset $(v,y)$, and $w$ will not claim asset $(v,w)$, so $v$
will not be able to claim asset $(z,v)$, as it will not have secret $s_x$.
In both cases, the outcome of $v$ is unacceptable.
\end{proof}


We now use the above properties to establish some global properties of $\protP$.
These properties involve  paths and cycles in $G$. The first corollary 
extends Corollary~\ref{cor: paths and cycles creation times}, and
is a direct consequence of Theorem~\ref{thm: local properties}(b).

\begin{corollary}\label{cor: paths and cycles timeouts}
If on some path each contract except possibly first is not protected by the seller,
then along this path the timeout values strictly decrease.
\end{corollary}


The next corollary follows
from Corollaries~\ref{cor: paths and cycles creation times} and~\ref{cor: paths and cycles timeouts}.

\begin{corollary}\label{cor: paths and cycles with same hashlock}
Let $P$ be a path such that all contracts on $P$ except possibly first are protected by a party $x$ 
that is not an internal vertex of $P$.
Then all contract creation times along $P$ strictly increase and all
timeout values strictly decrease.
\end{corollary}


\begin{corollary}\label{cor: paths and hashlocks}
{(a)}
Let $(u,v)$ be a contract protected by some party $x$ other than $v$.
Consider a path $P$ starting with arc $(u,v)$, that 
doesn't contain $x$ as an internal vertex and on which each contract is not protected by its seller.
Then all contracts along $P$ are protected by $x$.

{(b)}
Let $(u,v)$ be a contract protected by some party $x$ other than $u$.
Consider a path $P$ ending with arc $(u,v)$, that 
doesn't contain $x$ as an internal vertex and on which each contract is not protected by its buyer.
Then all contracts along $P$ are protected by $x$.
\end{corollary}

\begin{proof}
(a) The corollary follows easily by repeated application of Theorem~\ref{thm: local properties}.
Let $(v,w)$ be the second arc on $P$. By the assumption, $(v,w)$ is not protected by $v$, and
since $v$ has an incoming contract protected by $x$ and $x\neq v$, Theorem~\ref{thm: local properties} implies that
contract $(v,w)$ must be also protected by $x$. If $w=x$, this must be the end of $P$.
If $w\neq x$, then $w$ has an incoming contract protected by $x$, so we can repeat the same argument for $w$, and so on.
This implies part~(a).

(b) The proof for this part is symmetric to that for part~(a), with the only difference being that
we proceed now backwards from $u$ along $P$.
\end{proof}


\begin{theorem}\label{thm: simple paths to/from x}
{(a)} 
Let $P$ be a simple path starting at a vertex $x$ whose last contract is protected by $x$.
Then all contracts on $P$ are protected by $x$.

{(b)} 
Let $Q$ be a simple path ending at a vertex $x$ whose first contract is protected by $x$.
Then all contracts on $Q$ are protected by $x$.
\end{theorem}

\begin{proof}
\smallskip
(a)
Let $P = u_1,u_2,...,u_{p+1}$, where $u_1 = x$ and $(u_{p},u_{p+1})$ is protected by $x$. 
The proof is by contradiction. Suppose that $P$ violates part~(a), namely it
contains a contract not protected by $x$. (In particular, this means that $p\ge 2$.)
We can assume that among all simple paths that violate property~(a), $P$ is shortest. 
(Otherwise replace $P$ in the argument below by a shortest violating path.)
Then $(u_{p-1},u_p)$ is not protected by $x$, because otherwise the prefix of $P$
from $x$ to $u_p$ would be a violating path shorter than $P$.
So $(u_{p-1},u_p)$ is protected by $u_p$.
Using Lemma~\ref{lem: path contracts not protected by seller}, each contract
on the path $u_1,u_2,...,u_p$ is not protected by the seller.
Since $(u_p,u_{p+1})$ is protected by $x$ and $x\neq u_p$,
each contract on $P$ is not protected by its seller. 


\begin{figure}[ht]
\begin{center}
\includegraphics[valign=m,width = 3.75in]{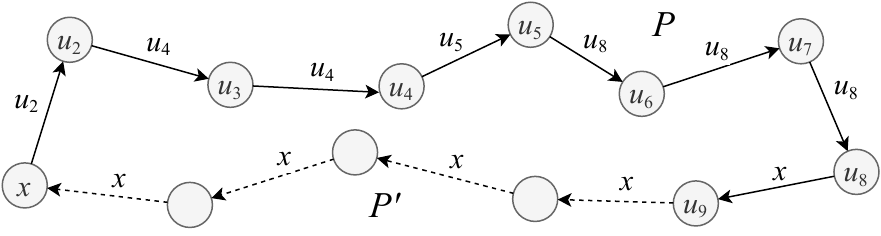}
\end{center}
\caption{Illustration of the proof of Theorem~\ref{thm: simple paths to/from x}(a).
Path $P$ is marked with solid arrows, path $P'$ is marked with dashed arrows.
Here, $p = 8$ and $u_1 = x$.
The labels on edges show the parties that protect them.
}\label{fig: protected_by_x_theorem_part_a}
\end{figure}


Next, we claim that there is 
a simple path $P'$ from $u_{p+1}$ to $x$ whose all contracts are protected by $x$.
If $u_{p+1} = x$, this is trivial, so assume that $u_{p+1}\neq x$.
Then, since $u_{p+1}$ has an incoming contract protected by $x$ and $x\neq u_{p+1}$, $u_{p+1}$ must have an
outgoing contract $(u_{p+1},w)$ protected by $x$. If $w =x$, we are done. Else, we repeat the process for $w$,
and so on. Eventually, extending this path we must end at $x$, for otherwise we would have a cycle with all
contracts protected by $x$ but not containing $x$, contradicting Corollary~\ref{cor: paths and cycles creation times}(b).
This proves that such path $P'$ exists. Since all contracts on $P'$ are protected by $x$,
they are not protected by their sellers. 

Finally, let $C$ be the cycle obtained by combining paths $P$ and $P'$. 
(See Figure~\ref{fig: protected_by_x_theorem_part_a}.) As shown
above, every contract on $C$ is not protected by its seller.
But this contradicts Corollary~\ref{cor: paths and cycles creation times}, completing 
the proof of part~(a).

\smallskip
(b)
Let $Q = v_0,v_1,...,v_q$, where $v_q = x$ and $(v_0,v_1)$ is protected by $x$.
Towards contradiction, suppose that some contract on $Q$ is not protected by $x$.
Analogously to part~(a) we can make two assumptions: One, we can assume that $q\ge 2$ and $(v_1,v_2)$ is not protected by $x$,
which implies that it is protected by $v_1$. Two, we can assume that there is a simple path $Q'$
from $v_1$ to $x$ with all contracts protected by $x$.

The rest of the proof is also similar to the proof of~(a), but with a twist (because we do not have
an analogue of Corollary~\ref{cor: paths and cycles creation times} for buyers).
The idea is to divide $Q$ into segments protected by the same party, for each segment take a reverse
path protected by this party, and combine these reverse paths into a path $T$ from $v_q$ to $v_1$, with each contract not protected by its seller.

We now describe this construction.
By the argument identical (except for reversing direction) to that in Lemma~\ref{lem: path contracts not protected by seller},
each contract $(v_i,v_{i+1})$, for $i \ge 1$, is protected by one of parties $v_1,v_2,...,v_i$.
Further, for each party $v_i$, where $i \notin \braced{0, q-1}$, 
the contracts protected by $v_i$ (if any) form a segment of $Q$ starting at $v_i$. 
Let $A$ consist of all indices $i$ for which contract $(v_i,v_{i+1})$ is protected by $v_i$, and also include $q$ in $A$.
Order $A$ in increasing order $a_1 < a_2 < ... < a_s$. So $a_1 = 1$, $a_s = q$, and for each $r = 1,2,...,s-1$
the segment $v_{a_r},v_{a_r+1},....,v_{a_{r+1}}$ of $Q$ has all contracts protected by $v_{a_r}$.


\begin{figure}[ht]
\begin{center}
\includegraphics[valign=m,width = 5in]{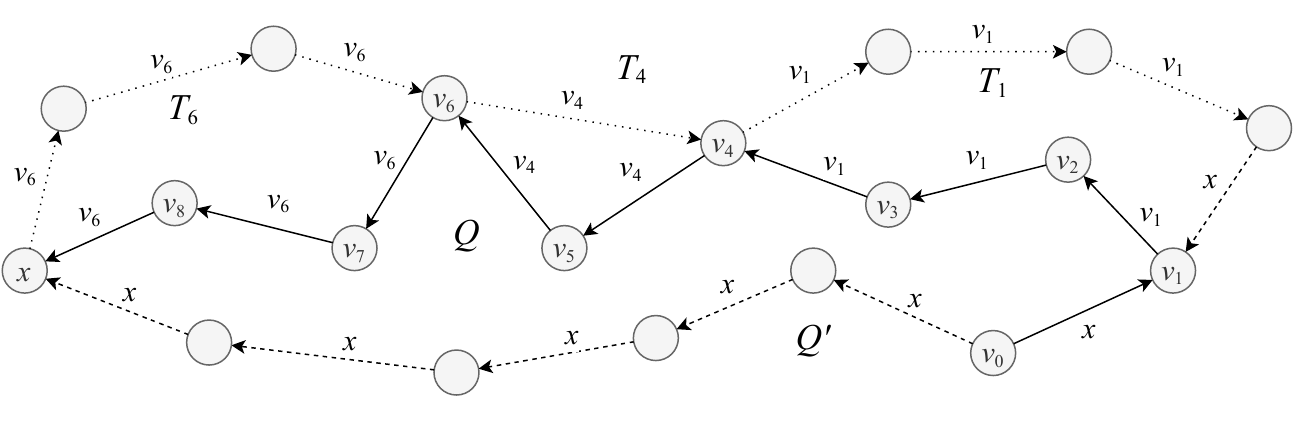}
\end{center}
\caption{Illustration of the proof of Theorem~\ref{thm: simple paths to/from x}(b).
Path $Q$ is marked with solid arrows. Here, $q=9$ and $v_9 = x$.
Path $Q'$ is marked with dashed arrows.
$Q$ is divided into three segments, protected by $v_1$, $v_4$, and $v_6$, respectively.
So $a_1 = 1$, $a_2 = 4$, $a_3 = 6$, and $a_4 = 9$.
The corresponding reverse paths are $T_1$, $T_4$ and $T_6$, marked with dotted arrows.
}\label{fig: protected_by_x_theorem_part_b}
\end{figure}


For each $r = 2,3,...,s$, contract $(v_{a_r-1},v_{r_a})$ is protected by $v_{r_a-1}$. 
By the same argument as in~(a), we then obtain that there is a simple path $T_{a_r-1}$ from $v_{r_a}$ to $v_{a_r-1}$
with all contracts protected by $v_{a_r-1}$, so none protected by its seller. The concatenation of these paths forms a path $T$
from $v_{r_s} = v_q$ to $v_{a_1} = v_1$.  So each contract on $T$ is not protected by its seller.

To finish, combine the two paths $T$ and $Q'$ into a cycle $C$.
On $C$, each contract is not protected by the seller.
This contradicts Corollary~\ref{cor: paths and cycles creation times}, completing the proof.
\end{proof}

%% file: 05.2_proof_of_characterization.tex


In this section we use the protocol properties established in the previous section
to prove necessary conditions for a digraph to admit an atomic
HTLC protocol. We start with protocols that use only one common hashlock for the whole graph.


\begin{lemma}\label{lem: one hashlock => bottleneck}
Suppose that $G$ has an atomic HTLC protocol $\protP$ in which only one party creates
a secret/hashlock pair.  Then $G$ must be a bottleneck graph.
\end{lemma}

\begin{proof}
Strong connectivity is a required property for \emph{any} atomic protocol to exist, see~\cite{Herlihy18}.
If $x$ is the party that created the secret/hashlock pair, then
Corollary~\ref{cor: paths and cycles creation times} implies that any cycle in $G$ must contain $x$.
Thus $x$ is a bottleneck of $G$, proving the theorem.
\end{proof}


We now consider the general case, when all parties are allowed to create
secret/hashlock pairs. The lemma below establishes the $(\Rightarrow)$ implication in Theorem~\ref{thm: main theorem}.

\begin{lemma}\label{lem: many hashlocks => reuniclus}
Suppose that $G$ has an atomic HTLC protocol $\protP$.
Then $G$ must be a reuniclus digraph.
\end{lemma}

\begin{proof}
Recall what we have established so far in Section~\ref{subsec: basic properties of shl-based protocols}.
From Theorem~\ref{thm: local properties}(c)
we know that, for each party $u$, all contracts involving $u$ as a party are 
protected either by $u$ or by just one other party. Using this property, we
define the control relation on parties, as follows:
If $u$ has any contracts protected by some other party $x$,
we will say that $x$ \emph{controls} $u$. Let $\controlgraph$ be a digraph whose vertices are
the parties that created secret/hashlock pairs,
and each arc $(u,x)$ represents the control relation, meaning that $x$ controls $u$.
We want to prove that $\controlgraph$ is a tree.

Each node in $\controlgraph$ has at most one outgoing arc. 
This property already implies that each connected component of $\controlgraph$ is a so-called \emph{1-tree},
namely a graph that has at most one cycle. So in order to show that $\controlgraph$ is actually a tree, it is
sufficient to show the two claims below.


\begin{claim}\label{cla: no cycles on control graph}
$\controlgraph$ does not have any cycles.
\end{claim}

We now prove Claim~\ref{cla: no cycles on control graph}. Towards contradiction,
suppose that $\controlgraph$ has a cycle, say $C = v_1,v_2,...,v_k,v_{k+1}$, where $v_{k+1} = v_1$.
Consider any arc $(v_i,v_{i+1})$ on $C$. This arc represents that $v_{i+1}$ controls $v_i$.
So, in $G$, $v_i$ has an outgoing contract protected by $v_{i+1}$.
Let $P_i$ be any path in $G$ starting with this contract and ending at $v_{i+1}$.
Then Theorem~\ref{thm: simple paths to/from x}(b) gives us that all contracts on $P_i$ protected by $v_{i+1}$.
Combining these paths $P_1,...,P_k$ we obtain a cycle $C'$ in $G$.
Then in $C'$, each contract is not protected by its seller, which would contract Corollary~\ref{cor: paths and cycles creation times}(b).
This completes the proof of Claim~\ref{cla: no cycles on control graph}.


\begin{claim}\label{cla: control graph connected}
$\controlgraph$ has only one tree.
\end{claim}

From Claim~\ref{cla: no cycles on control graph} 
we know that each (weakly) connected component of $\controlgraph$ is a tree. The roots of these trees
have the property that they are not protected by any other party.
To prove Claim~\ref{cla: control graph connected}, suppose towards contradiction that $\controlgraph$ has two
different trees, and denote by $r$ and $r'$ the roots of these trees. Since $r,r'$ are roots of
trees,  all contracts involving $r$ are protected by $r$ and all contracts involving $r'$ are protected by $r'$.
Consider any simple path $P = u_1,u_2,...,u_k$ from $u_1 = r$ to $r' = u_k$. 
Since the last contract on $P$ is protected by $r'$ and $r'$ is not in $\braced{u_1,u_2,...,u_{k-1}}$,
Lemma~\ref{lem: path contracts not protected by seller} implies that all contracts on this path
are not protected by their sellers. But this contradicts the fact that $(u_1,u_2)$ is protected by $u_1$.
This completes the proof of Claim~\ref{cla: control graph connected}.

\smallskip

We now continue with the proof of the theorem. Denote by $b_1,b_2,...,b_p$ the nodes of $\controlgraph$. 
For each $b_j$, define $G_j$ to be the subgraph induced by the contracts protected by $b_j$.
That is, for each contract $(u,v)$ protected by $b_j$ we add vertices $u$, $v$ and arc $(u,v)$ to $G_j$.
The necessary properties~(rg1) and~(rg2) follow from our results in Section~\ref{subsec: basic properties of shl-based protocols}.
It remains to show that subgraphs graphs $G_1,G_2,...,G_p$, together with tree $\controlgraph$,
satisfy conditions~(rg1) and~(rg2) that characterize reuniclus graphs.

Consider some $u\neq b_j$ that is in $G_j$. By the definition of $G_j$, $u$ is involved in a contract protected by $b_j$. 
Then Theorem~\ref{thm: local properties} gives us that $u$ has both an incoming and outgoing contract protected by $b_j$.
Take any path $P$ starting at an outgoing contract of $u$ protected by $b_j$ and ending at $b_j$.
Then Theorem~\ref{thm: simple paths to/from x}(b) implies that all contracts on $P$ are protected by $b_j$.
By the same reasoning, there is a path $P'$ starting at $b_j$ and ending with an incoming contract
of $u$ protected by $b_j$. Then, by Theorem~\ref{thm: simple paths to/from x}(a) all contracts on $P'$ are protected by $b_j$. 
This gives us that $G_j$ is strongly connected. 
And, by Corollary~\ref{cor: paths and cycles creation times},
$G_j$ cannot contain a cycle not including $b_j$. Therefore $G_j$ is a bottleneck graph with $b_j$ as its bottleneck.

We also need to prove that $G_j$ is in fact an induced subgraph, that is, if $u,v\in G_j$ and $G$
has an arc $(u,v)$, then $(u,v)\in G_j$ as well. That is, we need to prove that $(u,v)$ is protected by $b_j$.
Suppose, towards, contradiction, that $(u,v)$ is protected by some $b_i\neq b_j$.
Then both $u$ and $v$ are involved in contracts protected by both $b_i$ and $b_j$, and this implies that
$u = b_i$ and $v= b_j$, or vice versa. And this further implies that $b_i$ would be protected by $b_j$ and
vice versa, which would be a cycle in $\controlgraph$,
contradicting that $\controlgraph$ is a tree. This completes the proof of property~(rg1).

Finally, consider property~(rg2). If a vertex $u$ is not any of designated bottlebeck vertices $b_1,b_2,...,b_p$,
then, by Theorem~\ref{thm: local properties}, all its contracts are protected by the same party, which means
that it belongs to exactly one graph $G_j$.
On the other hand, if $u = b_j$, then again by Theorem~\ref{thm: local properties}, it 
is involved in only in contracts protected by itself and one other party, say $b_i$.
But then it belongs only to $G_j$ and $G_i$, completing the proof of~(rg2).
\end{proof}